\documentclass[lettersize,journal]{IEEEtran}
\usepackage{color}
\usepackage{amsmath,amsthm,amssymb,amsfonts}

\allowdisplaybreaks[4]

\usepackage[T1]{fontenc}
\usepackage{graphicx}
\usepackage{longtable}
\usepackage{float}
\usepackage{hyperref}

\usepackage{bm}

\usepackage[bottom]{footmisc}

\usepackage[mathscr]{eucal}

\usepackage{subfigure}

\usepackage{multirow}

\usepackage{cite}

\usepackage{epsfig}

\usepackage{soul}

\usepackage[lined,boxed,ruled]{algorithm2e}

\newcommand{\barjmath}{\bar{\jmath}}
\newtheorem{theorem}{Theorem}
\newtheorem{lemma}{Lemma}

\newtheorem{corol}{Corollary}

\newcommand{\funf}[1][\ba,\bA]{\mathbf{f}\left( #1\right)}

\newcommand{\Exp}{{\mathbb{E}}}

\newcommand{\braces}[1]{\left\lbrace #1\right\rbrace}

\newcommand{\setposi}[1]{\mathcal{Z}_{#1}^+}

\newcommand{\setnnega}[1]{\mathcal{Z}_{#1}}

\newcommand{\tr}[1]{\mathrm{tr}\left\lbrace #1\right\rbrace }

\newcommand{\diag}[1]{\mathrm{diag}\left\lbrace #1\right\rbrace }
\newcommand{\Diag}[1]{\mathrm{Diag}\left\lbrace #1\right\rbrace }

\newcommand{\rank}[1]{\mathrm{rank} \left(  #1 \right)  }

\newcommand{\pg}[3]{ p_{\mathrm{G}} \left(  #1;#2,#3  \right)  }

\newcommand{\toinf}[1]{#1 \to\infty}

\newcommand{\liminfty}[1]{\lim_{\toinf{#1}}}

\newcommand{\mtxvec}[1]{\mathrm{vec}\left\lbrace #1\right\rbrace }

\newcommand{\equaa}{\mathop{=}^{(\textrm{a})}}

\newcommand{\expb}[1]{\exp \left\lbrace  #1 \right\rbrace}

\newcommand{\ba}{\mathbf{a}}
\newcommand{\bb}{\mathbf{b}}
\newcommand{\bc}{\mathbf{c}}
\newcommand{\bd}{\mathbf{d}}
\newcommand{\be}{\mathbf{e}}

\newcommand{\bg}{\mathbf{g}}
\newcommand{\bh}{\mathbf{h}}

\newcommand{\bp}{\mathbf{p}}

\newcommand{\br}{\mathbf{r}}

\newcommand{\bt}{\mathbf{t}}

\newcommand{\bx}{\mathbf{x}}
\newcommand{\by}{\mathbf{y}}
\newcommand{\bz}{\mathbf{z}}

\newcommand{\bA}{\mathbf{A}}
\newcommand{\bB}{\mathbf{B}}

\newcommand{\bD}{\mathbf{D}}
\newcommand{\bE}{\mathbf{E}}
\newcommand{\bF}{\mathbf{F}}

\newcommand{\bH}{\mathbf{H}}
\newcommand{\bI}{\mathbf{I}}

\newcommand{\bK}{\mathbf{K}}

\newcommand{\bM}{\mathbf{M}}

\newcommand{\bP}{\mathbf{P}}
\newcommand{\bQ}{\mathbf{Q}}

\newcommand{\bT}{\mathbf{T}}
\newcommand{\bU}{\mathbf{U}}
\newcommand{\bV}{\mathbf{V}}

\newcommand{\bX}{\mathbf{X}}
\newcommand{\bY}{\mathbf{Y}}
\newcommand{\bZ}{\mathbf{Z}}

\newcommand{\bbC}{\mathbb{C}}
\newcommand{\bbR}{\mathbb{R}}

\newcommand{\bzero}{\mathbf{0}}
\newcommand{\bone}{\mathbf{1}}

\newcommand{\bSigma}{{\boldsymbol\Sigma}}

\newcommand{\bLambda}{{\boldsymbol\Lambda}}

\newcommand{\bOmega}{{\boldsymbol\Omega}}
\newcommand{\bomega}{{\boldsymbol\omega}}

\newcommand{\btheta}{{\boldsymbol\theta}}

\newcommand{\bgamma}{{\boldsymbol\gamma}}
\newcommand{\bmu}{{\boldsymbol\mu}}

\newcommand{\bxi}{{\boldsymbol\xi}}
\newcommand{\bvartheta}{{\boldsymbol\vartheta}}

\newcommand{\bbnu}{{\boldsymbol\nu}}

\newcommand{\dnnot}[2]{#1_{\textrm{#2}}}

\newcommand{\normmm}[1]{{\left\vert\kern-0.25ex\left\vert\kern-0.25ex\left\vert #1 
		\right\vert\kern-0.25ex\right\vert\kern-0.25ex\right\vert}}

\newcommand{\norm}[1]{\lVert #1 \rVert}

\definecolor{myback}{RGB}{204,232,207}

\begin{document}

\title{Simplified Information Geometry Approach for Massive MIMO-OFDM Channel Estimation - Part II: Convergence Analysis}	
\author{
	Jiyuan~Yang,~\IEEEmembership{Student~Member,~IEEE,}
	Yan~Chen,~
	Mingrui~Fan,
    ~Xiqi~Gao,~\IEEEmembership{Fellow,~IEEE,}
    Xiang-Gen~Xia,~\IEEEmembership{Fellow,~IEEE,}
	and 
	Dirk~Slock,~\IEEEmembership{Fellow,~IEEE}
}

\maketitle

\begin{abstract}
In Part II of this two-part paper, we prove the convergence of the simplified information geometry approach (SIGA) proposed in Part I.
For a general Bayesian inference
problem, we first show that the iteration of the common second-order natural parameter (SONP) is separated from that of the common first-order natural parameter (FONP).
Hence, the convergence of the common SONP can be checked independently.
We show that with the initialization satisfying a specific but large range, the common SONP is convergent regardless of the value of the damping factor. 
For the common FONP, we establish a sufficient condition of its convergence and prove that the convergence of the common FONP relies on the spectral radius of a particular matrix related to the damping factor.
We give the range of the damping factor that guarantees the convergence in the worst case. 
Further, we determine the range of the damping factor for massive MIMO-OFDM channel estimation by using the specific properties of the measurement matrices.
Simulation results are provided to confirm the theoretical results.
\end{abstract}

\section{Introduction} 
Numerous problems in signal processing eventually come to the issue of computing marginal probability density functions (PDF) from a high dimensional joint PDF.
In general, the calculation of direct marginalization could be unaffordable since operations such as matrix inversion could be involved.
In the past decades, many works have been devoted to providing an efficient way to compute the (approximate) marginal PDFs under various cases.
Among them, Bayesian inference approaches, e.g., message passing, Bethe free energy and etc, have attracted much interest due to their reliable performance and low computational complexities \cite{BP,jordan1999introduction,winn2005variational,1459044,malioutov2006walk,AMP,GAMP,EP,Bethe}.
Furthermore, besides the two advantages mentioned above, 
some of them possess favorable results in theory. 
\cite{malioutov2006walk} proposes the Gaussian belief propagation (BP) and shows that Gaussian BP is able to compute true marginal mean.
In \cite{AMP}, the powerful approximate message passing (AMP) algorithm is proposed.
It has been shown that AMP with Bayes-optimal denoiser can be treated as an exact approximation of loopy BP in the large system limit.
When the underlying factor graph is a tree, the expectation propagation (EP) in \cite{EP} is convergent and can exactly
achieve the Bayes-optimal performance.

Recently, we have introduced the information geometry approach (IGA) to the massive multiple-input multiple-output (MIMO) channel estimation \cite{IGA}.
We also improve the stability of IGA by introducing the damping factor and show that IGA can obtain accurate a posteriori mean at its fixed point.
On the basis of IGA, two  new results of IGA are revealed when the constant magnitude pilots are adopted.
Based on these new results, we propose a simplified IGA (SIGA) in Part I of this two-part paper \cite{SIGA}.
Although proposed for the massive MIMO-OFDM channel estimation, SIGA itself could serve as a generic Bayesian inference method which is suitable for Gaussian priors and constant magnitude measurement matrix.
In Part I, it has been shown that at the fixed point, the a posteriori mean obtained by SIGA is asymptotically optimal.
Furthermore, SIGA can be implemented efficiently when the measurement matrix has special structure.
For example, in massive MIMO channel estimation, the measurement matrix could be constructed by partial DFT matrices, SIGA can be then implemented by fast Fourier transform (FFT), which significantly reduces its computational complexity.

On the other hand, one standard limitation of Bayesian inference approaches is that they only work under the prerequisite that their parameters do converge.
However, due to the variety of problems, a unified way of proving convergence of Bayesian inference approaches has not yet been found.
Thus, the convergence needs to be proved for individual approaches as well as individual problems.
So far, there are only a few methods whose convergence has been relatively well revealed.
One sufficient condition named the walk-summability is proposed in \cite{malioutov2006walk} for the convergence of Gaussian BP.
Then, several extended conditions are proposed by the authors in  \cite{6872556,7004066}.
In \cite{9844780}, the convergence of orthogonal/vector AMP is proved based on the idea of "convergence in principle".
In \cite{8698290}, the convergence of AMP (or, equivalently, generalized AMP)  is proved in the special case with Gaussian priors.

Similar to most other Bayesian inference approaches, SIGA suffers from divergence.
However, by adding damping factor, its convergence can be significantly improved. 
This is an interesting observation, since in many iterative Bayesian inference approaches, such as, e.g., AMP, damped updating likewise plays an important role in convergence. 
In Part II of this two-part paper, we will give a theoretical analysis of the convergence of SIGA. 
The role of the damping factor in the iteration will also be clarified.
We summarize the main contributions as follows:
\begin{itemize}
	\item {We prove the convergence of the common  second-order natural parameter (SONP) in SIGA for a general Bayesian inference
	problem with the measurement matrix of
	constant magnitude property.}
	We show that the iteration of the common SONP is separated from that of the common first-order natural parameter (FONP).
	It is then proved that given the initialization satisfying a specific range, the common SONP is guaranteed to converge regardless of the value of damping factor, where the range of initialization is very large.
	\item We establish a sufficient condition of the convergence of the common FONP in SIGA.
	We show that the convergence of the common FONP depends on the spectral radius of the iterating system matrix.
	On this basis, we further give the range of the damping factor that guarantees the convergence of the common FONP in the worst case.
	\item {We apply the above general theories to the case of massive MIMO-OFDM channel estimation.} 
	Specifically, we determine a range of the damping factor that guarantees the convergence of the common FONP in SIGA from the properties of the measurement matrices.
\end{itemize}

The rest of Part II of this two-part paper is organized as follows. 
Preliminaries of SIGA are introduced, and
its convergence is analyzed in Section \uppercase\expandafter{\romannumeral2}.
A range of damping factors that guarantee the convergence of SIGA in massive MIMO-OFDM channel estimation is presented in Section \uppercase\expandafter{\romannumeral3}. 
Simulation results are provided in Section \uppercase\expandafter{\romannumeral4}. 
The conclusion is drawn in Section \uppercase\expandafter{\romannumeral5}.

Notations: 
The superscripts $\left( \cdot \right)^*$, $\left( \cdot \right)^T$ and $\left( \cdot \right)^H$ denote the conjugate, transpose and conjugate-transpose operator, respectively.
We use $\lceil x \rceil$ to denote the largest integer not larger than $x$.
$\bzero$ denotes zero vector or zero matrix with proper dimension when it fits.
$\Diag{\bx}$ denotes the diagonal matrix with $\bx$ along its main diagonal and $\diag{\bX}$ denotes a vector consisting of the diagonal elements of $\bX$.  
Define $\setposi{N} \triangleq \braces{1,2,\ldots,N}$ and $\setnnega{N} \triangleq \braces{0,1,\ldots,N}$. 
$y_n$, $a_{i,j}$ or $\left[\bA\right]_{i,j}$, and $\left[ \bA \right]_{:,i}$ denote the $n$-th component of the vector $\by$, the $\left(i,j\right)$-th component of the matrix $\bA$, and the $i$-th row of the matrix $\bA$, where the component indices start with $1$. 
$\ba < b$ means that each component in vector $\ba$ is smaller than the scalar $b$.
$\ba < \bc$ means that each
component in vector $\ba$ is smaller than the component in the corresponding position in vector $\bc$.
$\pg{\bh}{\bmu}{\bSigma}$ denotes the PDF of a complex Gaussian distribution $\mathcal{CN}\left( \bmu,\bSigma \right)$ for vector $\bh$ of complex random variables.

\section{{Convergence of SIGA}}
In this section, we first briefly introduce SIGA proposed in \cite{SIGA}, more details can be found therein.
Then, through re-expressing its iterations, we prove its convergence.

\subsection{SIGA}\label{sec:SIGA}
In Part II, we focus on the following Bayesian inference problem:
\begin{equation}\label{equ:rece signal model}
	\by = \bA\bh + \bz,
\end{equation}
where $\by \in \bbC^{N }$ is the observation, $\bA \in \bbC^{N \times M}$ is the deterministic (also known) measurement matrix,
$\bA$ satisfies constant magnitude property, i.e., $\left|a_{i,j}\right| = \left|a_{m,n}\right|, \forall i,j,m,n$, 
$\bh \sim \mathcal{CN}\left( \bzero,\bD \right)$ is the $M$-dimensional complex Gaussian random vector to be estimated, 
its covariance matrix $\bD$ is deterministic, known, positive definite and diagonal,
$\bz \sim \mathcal{CN}\left(\bzero,\sigma_z^2\bI\right)$ is the $N$-dimensional noise vector, and $\bh$ and $\bz$ are independent with each other.
Without loss of generality, assuming that
the components of $\bA$ have unit magnitude.
Our goal is to obtain the \textsl{a posteriori} information\footnote{Note that a virtual received signal model (\cite[Equation (28)]{SIGA}) is introduced such that the \textsl{a posteriori} mean obtained by SIGA can be asymptotically optimal.
The proof of convergence of SIGA is the same regardless of which received signal model is used. 
We consider only the received signal model \eqref{equ:rece signal model} for the sake of notational simplicity.} of $\bh$.
Given $\by$, the \textsl{a posteriori} distribution of $\bh$ is Gaussian. 
Thus, the \textsl{a posteriori} information of $\bh$ is determined by the \textsl{a posteriori} mean and the \textsl{a posteriori} covariance matrix.
We have $p\left( \bh|\by \right) = \pg{\bh}{\tilde{\bmu}}{\tilde{\bSigma}}$, where the \textsl{a posteriori} mean $\tilde{\bmu}$ and covariance matrix $\tilde{\bSigma}$ are given by   
\cite{kay1993fundamentals}
\begin{subequations}\label{equ:post information}
	\begin{equation}\label{equ:post mean}
		\tilde{\bmu} = \bD\left( \bA^H\bA\bD + \sigma_z^2\bI  \right)^{-1}\bA^H\by,
	\end{equation}
	\begin{equation}\label{equ:post covariance}
		\tilde{\bSigma} = \left( \bD^{-1} + \frac{1}{\sigma_z^2}\bA^H\bA \right)^{-1}.
	\end{equation}
\end{subequations} 
In the case of large $M$ and $N$, it is unaffordable to directly calculate $\tilde{\bmu}$ and $\tilde{\bSigma}$ since a large dimensional matrix inversion is involved and the complexity of it is $\mathcal{O}\left(M^3 + M^2N \right)$.
The aim of SIGA is calculating the approximations of the marginals of the \textsl{a posteriori} distribution, i.e., the approximations of  $p_i\left(h_i|\by\right), i\in \setposi{M}$.
Then, the \textsl{a posteriori} mean and variance can be obtained.
We begin with some essential definitions in SIGA. 
Given $\ba, \bb \in \bbC^{M}$, define a vector function as $\funf[\ba,\bb] \triangleq \left[ \ba^T, \bb^T \right]^T\in \bbC^{2M}$ and  an operator $\circ$ as $\ba \circ \bb \triangleq \frac{1}{2}\left(\bb^H\ba + \ba^H\bb\right)$.
Let $\bd =\funf[\bzero,\diag{-\bD^{-1}}] \in \bbC^{2M}$ and $\bt = \funf[\bh,\bh\odot\bh^*] \in \bbC^{2M}$.
Then, $p\left(\bh|\by\right)$ can be expressed as \cite{IGA,SIGA}
\begin{subequations}
	\begin{equation}\label{equ:a posteriori pdf 2}
		p\left(\bh|\by\right) = \expb{\bd \circ \bt + \sum_{n=1}^{N}c_n\left(\bh\right) - \psi_q},
	\end{equation}
	\begin{equation}\label{equ:c_n}
		c_n\left(\bh\right) \!=\! \frac{1}{\sigma_z^2}\left(  -\bh^H{\bgamma_n\bgamma_n^H}\bh + {y_n}\bh^H\bgamma_n + {y_n^*}\bgamma_n^H\bh\right),
	\end{equation}
		\begin{equation}\label{equ:gamma_n}
		\bgamma_n = \left[ \bA^H \right]_{:,n} = \left[ a_{n,1}^*, \ \cdots, \  a_{n,M}^* \right]^T \in \bbC^{M},
	\end{equation}
\end{subequations}
where $\psi_q$ is the normalization factor. 
In \eqref{equ:a posteriori pdf 2},  $\bt$ only contains the  statistics of single random variables, i.e., $h_i$ and $\left|h_i\right|^2, i\in \mathcal{Z}_M^+$, and all the interactions (cross terms), $h_ih_j^*, i\neq j$, are contained in $c_n\left(\bh\right), n \in \mathcal{Z}_N^+$. 
SIGA is to approximate each $c_n\left(\bh\right)$ as  $\bxi_{n}\circ\bt$ in an iterative manner, where $\bxi_n \in \bbC^{M}$ is referred as to the approximation item.
In this way, we have
\begin{equation}
	p\left(\bh|\by\right) \approx p_0\left(\bh;\bvartheta_0\right) = \expb{\left(\bd + \bvartheta_0\right) \circ \bt - \psi_0 },
\end{equation}
where $\bvartheta_{0} = \sum_{n=1}^{N}\bxi_{n} \in \bbC^{2M}$ and
$\psi_0$ is the normalization factor.
The marginals of $p_0\left(\bh;\bvartheta_0\right)$ can be calculated directly since it contains no cross terms between random variables. 
To obtain $\bxi_{n}, n\in \setposi{N}$, and
$\bvartheta_{0}$, SIGA constructs the following three types of manifolds: the original manifold (OM), the objective manifold (OBM) and the auxiliary manifold (AM).
The OM is defined as the set of PDFs of $M$-dimensional complex Gaussian random vectors,
\begin{equation}
	\mathcal{M}_{\textrm{or}} = \left\lbrace p\left(\bh\right) \!=\! p_G\left(\bh;\bmu,\bSigma\right), \bmu\in\bbC^{M}, \bSigma\in \mathbb{H}_+^M \right\rbrace, 
\end{equation}
where  $\mathbb{H}_+^M$ is the set of $M$ dimensional positive definite matrices.  The OBM is defined as
\begin{equation}\label{equ:definition of M0}
	\mathcal{M}_0 = \left\lbrace p_0\left(\bh;\bvartheta_0\right) = \expb{\left(\bd+\bvartheta_0\right)\circ\bt - \psi_0\left(\bvartheta_0\right)}  \right\rbrace, 
\end{equation} 
where $\bvartheta_0 = \funf[\btheta_0,\bbnu_0]$,  
$\btheta_0 \in \bbC^{M}$ and $\bbnu_0 \in \bbR^{M}$ are referred to as the natural parameter (NP),
the first-order natural parameter (FONP) and the second-order natural parameter (SONP) of $p_0$, respectively, {and $\psi_0\left(\bvartheta_0\right)$ is the free energy (normalization factor).}
$N$ AMs are defined, where the $n$-th of them is given by
\begin{subequations}\label{equ:definition of Mn}
	\begin{equation}
		\mathcal{M}_n = \braces{p_n\left(\bh;\bvartheta\right)}, n\in\mathcal{Z}_N^+,
	\end{equation}
	\begin{equation}
		p_n\left(\bh;\bvartheta\right) \!=\! \expb{\left(\bd + \bvartheta\right)\circ \bt + c_n\left(\bh\right) - \psi_n\left(\bvartheta\right)},
	\end{equation}
\end{subequations}
where $\bvartheta = \funf[\btheta,\bbnu]$, $\btheta \in \bbC^{M}$ and $\bbnu \in \bbR^{M}$ are referred to as the common NP, the common FONP and the common SONP of $p_n$, respectively,
and $\psi_n\left(\bvartheta\right)$ is the free energy.
We call $\bvartheta$ the common NP since that all the $N$ AMs share the same $\bvartheta$.
The OBM and the AMs are submanifolds of the OM since the distributions are all $M$-dimensional complex Gaussian distributions, but with different constraints \cite{SIGA}.

We now introduce the iteration of the SIGA.
Mathematically, the iteration of SIGA can be summarized as using the value of $\bvartheta$ at the $t$-th iteration (denoted as $\bvartheta\left(t\right)$) to compute its value at the $\left(t+1\right)$-th iteration, i.e., $\bvartheta\left(t+1\right)$, until convergence.
After $\bvartheta$ is converged, we calculate $\bvartheta_{0}$ as $\bvartheta_{0} = \frac{N}{N-1}\bvartheta$.
Then,  $p\left(\bh;\bvartheta_0\right)$ in \eqref{equ:definition of M0} is referred as to the approximation of the product of the marginals of the \textsl{a posreriori} distributions, i.e., $\prod_{i=1}^{M}p_i\left(h_i|\by\right)$.
And the \textsl{a posteriori} mean and variance are given by $\bmu_0\left(\bvartheta_{0}\right)$ and the diagonal of $\bSigma_{0}\left(\bvartheta_{0}\right)$, respectively, where 
\begin{subequations}\label{equ:mu_0 and Sigma_0}
	\begin{equation}\label{equ:mu_0}
		\bmu_0\left(\bvartheta_{0}\right) = \frac{1}{2}\bSigma_{0}\left(\bvartheta_{0}\right)\btheta_{0},
	\end{equation}
	\begin{equation}\label{equ:Sigma_0}
		\bSigma_{0}\left(\bvartheta_{0}\right) = \left( \bD^{-1}-\Diag{\bbnu_0} \right)^{-1},
	\end{equation}
\end{subequations}
and both $\bmu_0\left(\bvartheta_{0}\right)$ and  $\bSigma_{0}\left(\bvartheta_{0}\right)$ are determined by vector $\bvartheta_{0}$.
Specifically, given $\bvartheta\left(t\right) = \funf[\btheta\left(t\right),\bbnu\left(t\right)]$ at the $t$-th iteration, $\bvartheta\left(t+1\right) = \funf[\btheta\left(t+1\right),\bbnu\left(t+1\right)]$ is then calculated as \eqref{equ:tilde vartheta_s} \cite{SIGA}, where  $0<d\le 1$ is the damping factor, and $\bLambda\left(\bbnu\left(t\right)\right)$ and $\beta\left(\bbnu\left(t\right)\right)$ are given by
\begin{subequations}\label{equ:Lambda and beta}
	\begin{equation}\label{equ:Lambda}
		\bLambda\left( \bbnu\left(t\right) \right) = \left( \bD^{-1} - \Diag{\bbnu\left(t\right)}\right)^{-1},
	\end{equation}
	\begin{equation}\label{equ:beta}
		\beta\left(\bbnu\left(t\right)\right) = {\sigma}_z^2 + \tr{\bLambda\left(\bbnu\left(t\right)\right)}.
	\end{equation}
\end{subequations}
Briefly, $\bvartheta\left(t+1\right)$ is obtained by the $m$-projections of $p_n\left(\bh;\bvartheta\left(t\right)\right), n\in \setposi{N}$, onto the OBM.
The detailed calculation can be found in the Sec. \uppercase\expandafter{\romannumeral4} of \cite{SIGA}.
We summarize the process of the  SIGA in Algorithm \ref{Alg:SIGA}. 
Comparing to \cite[Algorithm 2]{SIGA}, we fix the range of  $\bbnu$ in the initialization, i.e., 
\begin{equation}
	-\frac{N-1}{\sigma_z^2}\mathbf{1}\le \bbnu\left(0\right) \le \bzero,
\end{equation}
which will be explained in detail below.
We shall see that this new range of $\bbnu$ will ensure the convergence of SIGA.

\begin{algorithm}[htbp]
	\SetAlgoNoLine 
	\caption{SIGA}
	\label{Alg:SIGA}
	
	\KwIn{The covariance $\bD$ of the a priori distribution $p\left(\bh\right)$, the received signal $\by$, the noise power $\sigma_z^2$ and the maximal iteration number $t_{\mathrm{max}}$.}
	
	\textbf{Initialization:} set $t=0$, set damping $d$, where $0< d\le 1$, initialize the common NP as $\bvartheta\left(0\right) = \funf[\btheta\left(0\right),\bbnu\left(0\right)]$ and ensure $ -\frac{N-1}{\sigma_z^2}\mathbf{1}\le \bbnu\left(0\right)\le \bzero$; 
	
	\Repeat{\rm{Convergence or $t > t_{\mathrm{max}}$}}{
		1. Update $\bvartheta = \funf[\btheta,\bbnu]$ as \eqref{equ:tilde vartheta_s}, where $\bLambda\left(\bbnu\left(t\right)\right)$ and $\beta\left(\bbnu\left(t\right)\right)$ are given by \eqref{equ:Lambda} and \eqref{equ:beta}, respectively;\\
		2. $t = t+1$;}
	
	\KwOut{\rm{Calculate the NP of $p_0\left(\bh;\bvartheta_{0}\right)$ as $\bvartheta_{0} = \frac{N}{N-1}\bvartheta\left(t\right)$.The mean and variance of the approximate marginal, $p_i\left( h_i|\by\right)$, $i\in \setposi{M}$, are given by the $i$-th component of $\bmu_0$ and $\diag{\bSigma_{0}}$, respectively, where $\bmu_0$ and $\bSigma_{0}$ are calculated by  \eqref{equ:mu_0 and Sigma_0}. }}
\end{algorithm}

\begin{figure*}
	\begin{subequations}\label{equ:tilde vartheta_s}
		\begin{equation}\label{equ:tilde theta_s}
			\btheta\left(t+1\right) = 
			\frac{d\left(N-1\right)}{N}\left( \bI - \frac{1}{\beta\left(\bbnu\left(t\right)\right)}\bLambda\left(\bbnu\left(t\right)\right) \right)^{-1}\left[ \frac{1}{\beta\left(\bbnu\left(t\right)\right)}\bA^H\left(2\by - \bA\bLambda\left(\bbnu\left(t\right)\right){\btheta}\left(t\right) \right) + N{\btheta}\left(t\right) \right] + \left(1-dN\right)\btheta\left(t\right)
		\end{equation}
		\begin{equation}\label{equ:tilde nu_s}
			\bbnu\left(t+1\right) = {d\left(N-1\right)}\diag{\bD^{-1}-\left( \bLambda\left(\bbnu\left(t\right)\right) - \frac{1}{\beta\left(\bbnu\left(t\right)\right)}\bLambda^2\left(\bbnu\left(t\right)\right)\right)^{-1}} + \left(1-dN\right)\bbnu\left(t\right)
		\end{equation}
	\end{subequations}
	\hrule
\end{figure*}


\subsection{Convergence}
From Algorithm \ref{Alg:SIGA}, we can find that to prove the convergence of SIGA, we only need to prove the convergence of $\bvartheta\left(t\right)$.
Given $\bvartheta\left(t\right) = \funf[\btheta\left(t\right),\bbnu\left(t\right)]$ at the $t$-th time,
we re-express $\bvartheta\left(t+1\right) = \funf[\btheta\left(t+1\right),\bbnu\left(t+1\right)]$ in \eqref{equ:tilde vartheta_s} as the follows:
\begin{subequations}\label{equ:update of SIGA}
	    \begin{equation}\label{equ:update of nu_0}
		\bbnu\left(t+1\right) = \tilde{\bg}\left(\bbnu\left(t\right)\right) \triangleq d \bg\left(\bbnu\left(t\right)\right) + \left(1-d\right)\bbnu\left(t\right),
	\end{equation}
	\begin{equation}\label{equ:function g text}
		\bg\left( \bbnu\left(t\right) \right) = -\left(N-1\right)\diag{ \left(  \beta\left(\bbnu\left(t\right)\right)\bI - \bLambda\left(\bbnu\left(t\right)\right) \right)^{-1}  },
	\end{equation}
\end{subequations}
\begin{subequations}\label{equ:intermideate variable}
			\begin{equation}\label{equ:update of theta_0}
		\btheta\left(t+1\right) 
		= \tilde{\bB}\left(\bbnu\left(t\right)\right)\btheta\left(t\right) + \bb\left(\bbnu\left(t\right)\right),
	\end{equation}
		\begin{equation}\label{equ:tilde B}
		\tilde{\bB}\left(\bbnu\left(t\right)\right) \triangleq d\bB\left(\bbnu\left(t\right)\right) + \left(1-d\right)\bI,
	\end{equation}
	\begin{equation}\label{equ:B}
		\begin{split}
			\bB\left(\bbnu\left(t\right)\right) = &\frac{N-1}{\beta\left(\bbnu\left(t\right)\right)}\left( \bI - \frac{1}{\beta\left(\bbnu\left(t\right)\right)}\bLambda\left(\bbnu\left(t\right)\right) \right)^{-1} \\
			&\ \ \times\left(\bI - \frac{1}{N}\bA^H\bA\right)\bLambda\left(\bbnu\left(t\right)\right),	
		\end{split}
	\end{equation}
	\begin{equation}
		\bb\left(\bbnu\left(t\right)\right) = \frac{2d\left(N-1\right)}{N\beta\left(\bbnu\left(t\right)\right)}\left(\bI - \frac{1}{\beta\left(\bbnu\left(t\right)\right)}\bLambda\left(\bbnu\left(t\right)\right)\right)^{-1}\bA^H\by,
	\end{equation}
\end{subequations}
where \eqref{equ:update of SIGA} is the iterating system of $\bbnu$, \eqref{equ:intermideate variable} is the iterating system of $\btheta$, and the derivations are given in Appendix \ref{proof:calculation of update}.
All the above matrices that need to be inverted are shown to be invertible at each iteration in Appendix \ref{proof:calculation of update}, which guarantees the iterating systems defined by \eqref{equ:update of SIGA} and \eqref{equ:intermideate variable} are valid,
$\tilde{\bB}$ and $\bB$ are two matrix functions with $\bbnu\left(t\right)$ being the variable, i.e., $\tilde{\bB},\bB:\bbR^{M} \to \bbC^{M\times M}$,
and $\bb$, $\tilde{\bg}$ and $\bg$ are three vector functions with $\bbnu\left(t\right)$ being the variable, i.e., $\bb:\bbR^{M} \to \bbC^{M}$, 
and $ \tilde{\bg}, \bg:\bbR^{M} \to \bbR^{M}$. 
Also, we have the following lemma.
\begin{lemma}\label{prop:0}
	Given a finite initialization  $\bvartheta\left(0\right) = \funf[\btheta\left(0\right),\bbnu\left(0\right)]$ with $-\frac{N-1}{\sigma_z^2}\mathbf{1}\le  \bbnu\left(0\right) \le  \bzero$, then at each iteration,  $\bvartheta\left(t\right) = \funf[\btheta\left(t\right),\bbnu\left(t\right)]$ satisfies: $\btheta\left(t\right)$ and $\bbnu\left(t\right)$ are finite, and $\bbnu\left(t\right) \le  \bzero$.
	Specifically, we have $\bbnu\left(0\right) \le \bzero$ and $\bbnu\left(t\right) < \bzero, t \ge 1$.
\end{lemma}
\begin{proof}
	See in Appendix \ref{proof:calculation of update}.
\end{proof}
We refer matrix $\tilde{\bB}$ in \eqref{equ:tilde B} as the iterating system matrix of $\btheta$, which is determined by the common SONP $\bbnu$ and the measurement matrix $\bA$ at each iteration.
Combining \eqref{equ:update of nu_0} and \eqref{equ:update of theta_0}, we can see that $\bbnu\left(t+1\right)$ only depends on $\bbnu\left(t\right)$ and does not depend on $\btheta\left(t\right)$, 
while $\btheta\left(t+1\right)$ depends on both $\btheta\left(t\right)$ and $\bbnu\left(t\right)$. 
This shows that the iterating system of $\bbnu$ is separated from that of $\btheta$, and hence, the convergence of $\bbnu\left(t\right)$ can be checked individually.
We then consider the convergence of $\bbnu\left(t\right)$.
To this end, we first present the following lemma about the function $\tilde{\bg}\left(\bbnu\right)$ defined in \eqref{equ:update of nu_0}.

\begin{lemma}\label{lemma:function tilde{g}}
	Given $ \bbnu \le \bzero$, $\tilde{\bg}\left(\bbnu\right)$ satisfies the following two properties.\\
	1. Monotonicity: If $\bbnu < \bbnu' \le \bzero$, then $\tilde{\bg}\left( \bbnu\right) < \tilde{\bg}\left(\bbnu'\right)$.\\
	2. Scalability: Given a positive constant $0<\alpha<1$, we have $\tilde{\bg}(\alpha\bbnu)<\alpha\tilde{\bg}(\bbnu)$.\\
	Moreover, if $\tilde{\bg}_{\textrm{min}} \le \bbnu \le \bzero$ with $\tilde{\bg}_{\textrm{min}} \triangleq -\frac{N-1}{\sigma_z^2}\mathbf{1} \in \bbR^M$, we have $\tilde{\bg}_{\textrm{min}} < \tilde{\bg}\left(\bbnu\right) < \bzero$.
\end{lemma}
\begin{proof}
	See in Appendix \ref{proof:function tilde{g}}.
\end{proof}

Based on Lemma \ref{lemma:function tilde{g}}, we then have the following theorem. 
\begin{theorem}\label{The-theta_1 conv}
	Given initialization $\bbnu\left(0\right)$ with $ \tilde{\bg}_{\textrm{min}}\le  \bbnu\left(0\right) \le  \bzero$, where $\tilde{\bg}_{\textrm{min}}$ is defined in Lemma \ref{lemma:function tilde{g}}, the sequence $\bbnu\left(t+1\right) = \tilde{\bg}\left( \bbnu\left(t\right) \right)$ converges to a finite fixed point $\bbnu^\star$, where $\tilde{\bg}_{\textrm{min}}< \bbnu^{\star} < \bzero$.
\end{theorem}
\begin{proof}
	See in Appendix \ref{proof:theta_1 conv}.
\end{proof}
From Theorem \ref{The-theta_1 conv}, we can find that $\bbnu\left(t\right)$ converges to a finite fixed point as long as the initialization satisfies $ \tilde{\bg}_{\textrm{min}}\le \bbnu\left(0\right) \le \bzero$, and this range can be quite large.
For example, in the simulations of \cite{SIGA}, $N = 46080$, and when the noise variance is $\sigma_z^2 = 1$, then we obtain $\tilde{\bg}_{\textrm{min}} = -46079\times \bone$, which is a quite small negative vector. 
Theorem \ref{The-theta_1 conv} also shows that the convergence of $\bbnu\left(t\right)$ is not related to the damping factor $d$.
Yet we shall see that the convergence of $\btheta\left(t\right)$ is related to the choice of the damping factor later.
Define 
\begin{equation}\label{equ:Lambda0 star}
	\bLambda^\star = \left( \bD^{-1} - \Diag{\bbnu^\star} \right)^{-1},
\end{equation}
\begin{equation}\label{equ:beta0 star}
	\beta^\star = \sigma_z^2 + \tr{\bLambda^\star}.
\end{equation}
From Theorem \ref{The-theta_1 conv}, diagonal matrix $\bLambda^\star$ is positive definite and $\beta^\star > 0$.
From \eqref{equ:update of nu_0} and ${\bbnu^\star} = \tilde{\bg}\left(\bbnu^\star\right)$, we have $\bbnu^\star = \bg\left(\bbnu^\star\right)$.
Then, from the first equation of \eqref{equ:aux1 in AppA} in Appendix \ref{proof:calculation of update}, we have
\begin{equation}
	\frac{N}{N-1}\bbnu^\star = \diag{\bD^{-1} - \left( \bLambda^\star - \frac{1}{\beta^\star}\left(\bLambda^\star\right)^2 \right)^{-1}},
\end{equation}
which implies that
\begin{equation}\label{equ:condition of nu}
	\bLambda^\star - \frac{1}{\beta^\star}\left(\bLambda^\star\right)^2 = \left(\bD^{-1} - \frac{N}{N-1}\Diag{\bbnu^\star}   \right)^{-1}.
\end{equation}

Define 
\begin{equation}\label{equ:tilde B 0}
	\tilde{\bB}^{\star} = \tilde{\bB}\left( \bbnu^{\star} \right) = d\bB^{\star} + \left(1-d\right)\bI,
\end{equation}
where $\bB^{\star} = {\bB}\left( \bbnu^{\star} \right)$ and 
$\bb^\star = \bb\left( \bbnu^{\star} \right)$.
From the definition, $\tilde{\bB}^\star$ is determined by the fixed point of the common SONP $\bbnu^\star$ and the measurement matrix $\bA$, which does not vary with iterations.
To avoid ambiguity,  the iterating system matrix refers in particular to $\tilde{\bB}^\star$ in the rest of the paper, since the convergence condition for the iterating system of $\btheta\left(t\right)$ given in the following lemma depends only on the spectral radius of $\tilde{\bB}^\star$.
\begin{lemma}\label{The-synchronous updates} 
	Given a finite initialization $\btheta\left(0\right) \in \bbC^{M \times 1}$ and $\bbnu\left(0\right)$ with $-\frac{N-1}{\sigma_z^2}\mathbf{1}\le   \bbnu\left(0\right) \le  \bzero$.
	Then,
	$\btheta\left(t\right)$ in \eqref{equ:intermideate variable} converges to its fixed point  if the spectral radius of $\tilde{\bB}^\star$ is less than $1$, i.e., $\rho\left(\tilde{\bB}^{\star}\right) < 1$, with $\rho\left(\tilde{\bB^\star}\right) = \max\braces{\left|\lambda\right|: \lambda \  \textrm{is an eigenvalue of } \tilde{\bB^\star}}$.  
\end{lemma} 
\begin{proof}
See in Appendix \ref{Proof-The-syn}.
\end{proof}

From Lemma \ref{The-synchronous updates}, we see that when $\bbnu$ converges and the spectral radius of the iterating system matrix in \eqref{equ:tilde B 0} is less than $1$, $\btheta$ converges. 
We next give an analysis of the eigenvalue distribution of $\tilde{\bB}^\star$ and a theoretical explanation for the improved convergence of $\btheta$ under damped updating.
We begin with the eigenvalues  of $\bB^\star$ in \eqref{equ:tilde B 0}.
As mentioned above, 
when $\bbnu$ converges to $\bbnu^{\star}$, from \eqref{equ:condition of nu} and \eqref{equ:tilde B 0} it is not difficult to obtain 
\begin{equation}\label{equ:tilde B star 1}
	\begin{split}
		\bB^\star &= \frac{N-1}{\beta^{\star}}\left( \bI - \frac{1}{\beta^{\star}}\bLambda^\star \right)^{-1}
		\left(\bI -\frac{1}{N}\bA^H\bA\right)\bLambda^\star \\
		&=	 \left( \bI - \frac{1}{N}\bD^{-1}\bLambda^{\star} \right)\left(N\bI - \bA^H\bA\right)\left( \frac{1}{\beta^{\star}}\bLambda^{\star} \right).
	\end{split}
\end{equation}
We can find that ${\bB}^\star$ is the product of three matrices. 
The first matrix of the right hand side of \eqref{equ:tilde B star 1} satisfies the following property.
\begin{lemma}\label{prop:1}
	$\bI - \frac{1}{N}\bD^{-1}\bLambda^{\star}$ is diagonal with diagonal components all  positive and smaller than $1$. 
\end{lemma}
\begin{proof}
	From \eqref{equ:Lambda0 star}, we can obtain that $\bzero<\diag{\bD^{-1}\bLambda^{\star}} < \mathbf{1}$, which implies that the diagonal of $\bI - \frac{1}{N}\bD^{-1}\bLambda^{\star}$ is positive and smaller than $1$.
	This completes the proof.
\end{proof}

%

Since all the three matrices in the product in \eqref{equ:tilde B star 1} are Hermitian, we have \cite[Exercise below Theorem 5.6.9]{horn2012matrix}
\begin{equation}\label{equ:range of eigen of B}
	\begin{split}
		&\rho\left(\bB^{\star}\right) \\
		\le &\rho\left( \bI - \frac{1}{N}\bD^{-1}\bLambda^{\star} \right)\rho\left(N\bI - \bA^H\bA\right)\rho\left( \frac{1}{\beta^{\star}}\bLambda^{\star} \right).
	\end{split}
\end{equation}
From Lemma \ref{prop:1}, we can obtain that
\begin{equation}\label{equ:aux1 in text}
	\rho\left( \bI - \frac{1}{N}\bD^{-1}\bLambda^{\star}  \right) <1.
\end{equation} 
%
%
\begin{lemma}\label{Lemma-spectral radius}
	The spectral radius of $\bLambda^*$ satisfies 
	\begin{equation}
		\rho\left( \bLambda^\star \right) < \frac{\beta^\star}{N}.
	\end{equation}
\end{lemma}
\begin{proof}
	See in Appendix \ref{{Proof-Lemma-spec}}.
\end{proof}

We next show some properties of the eigenvalues  of $\bB^{\star}$ and $\tilde{\bB}^\star$.
\begin{lemma}\label{lemma:eigens of B star}
	Denote the eigenvalues of $\bB^{\star}$ as $\lambda_{B,i}, i\in \setposi{M}$.
	Then,  $\braces{\lambda_{B,i}}_{i=1}^M$ are all real and
	\begin{equation}\label{equ:range of v_B}
		-\frac{\rho\left(N\bI - \bA^H\bA\right)}{N} < \lambda_{B,i} <1.
	\end{equation}
\end{lemma}
\begin{proof}
	See in Appendix \ref{proof:eigens of B star}.
\end{proof}

Denote the eigenvalues of the iterating system matrix $\tilde{\bB}^{\star}$ in \eqref{equ:tilde B 0} as $\tilde{\lambda}_i, i \in \setposi{M}$.
From \eqref{equ:tilde B 0}, we have $\tilde{\lambda}_i = d\lambda_{B,i} + 1-d, i \in \setposi{M}$.
We then have the following lemma.
\begin{lemma}\label{lemma:eigenvalues of tilde B}
	The eigenvalues of $\tilde{\bB}^\star$ are all real and satisfy
	\begin{equation}
		1-d\left( 1+ \frac{\rho\left(N\bI - \bA^H\bA\right)}{N} \right) < \tilde{\lambda}_i <1.
	\end{equation}
\end{lemma}
\begin{proof}
	This is a direct result from Lemma \ref{lemma:eigens of B star}.
\end{proof}

Combining Lemmas \ref{lemma:eigens of B star} and \ref{lemma:eigenvalues of tilde B}, we have the following theorem.
\begin{theorem}\label{corol:rang of v'}
	Given a finite initialization $\btheta\left(0\right) \in \bbC^{M \times 1}$ and $\bbnu\left(0\right)$ with $-\frac{N-1}{\sigma_z^2}\mathbf{1}\le   \bbnu\left(0\right) \le  \bzero$.
	Then,
	$\btheta\left(t\right)$ in \eqref{equ:intermideate variable} converges to its fixed point  if the damping factor satisfies
	\begin{equation}\label{equ:range of d 0}
		d < \frac{2}{ 1+\frac{\rho\left(N\bI - \bA^H\bA\right)}{N}  }.
	\end{equation}
\end{theorem}
\begin{proof}
	This is a direct result from Lemmas \ref{The-synchronous updates} and \ref{lemma:eigenvalues of tilde B}.
\end{proof}

From Theorem \ref{corol:rang of v'}, we can find that SIGA will always converge with a sufficiently small damping factor  and the range of $d$ is mainly determined by $\rho\left(N\bI - \bA^H\bA\right)$. 
The spectral radius $\rho\left(N\bI - \bA^H\bA\right)$ depends on the measurement matrix $\bA$.
We next discuss the range of $\rho\left(N\bI - \bA^H\bA\right)$ in the worst case and give the range of damping factor accordingly.
The range of $\rho\left(N\bI - \bA^H\bA\right)$ and the corresponding range of damping factor in massive MIMO-OFDM channel estimation will be discussed in the next section.
\begin{theorem}\label{lemma:radius of NI-A^HA}
	The spectral radius of $N\bI - \bA^H\bA$ satisfies 
	\begin{equation}
		\rho\left(N\bI- \bA^H\bA \right)\le NM -N.
	\end{equation}
	If $\rank{\bA} = 1$, then $\rho\left(N\bI- \bA^H\bA \right) = NM-N$.
\end{theorem}
\begin{proof}
	See in Appendix \ref{proof:radius of NI-A^HA}.
\end{proof}
\begin{corol}\label{corol:1}
Given a finite initialization $\btheta\left(0\right) \in \bbC^{M \times 1}$ and $\bbnu\left(0\right)$ with $-\frac{N-1}{\sigma_z^2}\mathbf{1}\le   \bbnu\left(0\right) \le  \bzero$.
Then,
$\btheta\left(t\right)$ in \eqref{equ:intermideate variable} converges to its fixed point  if the damping factor satisfies $d < \frac{2}{M}$.
\end{corol}
\begin{proof}
	It is a direct result from Theorems \ref{corol:rang of v'} and \ref{lemma:radius of NI-A^HA}.
\end{proof}

From Corollary \ref{corol:1}, we can find that in the worst case, if $d< \frac{2}{M}$, then SIGA converges.

\section{Application to Massive MIMO-OFDM Channel Estimation}

In this section, we will discuss the range of $\rho\left(N\bI - \bA^H\bA\right)$ in massive MIMO-OFDM channel estimation, where the range of $d$ can be expanded.
We first consider the case where general pilot sequences with constant magnitude
property are adopted.
In this case, $\bA$ is given in \cite[Equation (8)]{SIGA} that is briefly described below. 
Let us first briefly introduce the system configuration and some notations in \cite{SIGA}.
Consider the following uplink massive MIMO-OFDM channel estimation problem: A base station equipped with $N_r = N_{r,v}\times N_{r,h}$ uniform planar array (UPA) serves $K$ single antenna users, where $N_{r,v}$ and $N_{r,h}$ are the numbers of the antennas at each vertical column and horizontal row, respectively.
The number of subcarriers and cyclic prefix  (CP) length of OFDM modulation are $N_c$ and $N_g$, respectively.
The number of subcarriers used for training is $N_p \le N_c$. 
Let $\bY \in \bbC^{N_r\times N_p}$  and $\bZ \in \bbC^{N_r \times N_p}$ be the space-frequency domain received signal and noise, respectively, 
then, we have the following received signal model \cite{SIGA}
\begin{equation}\label{rece signal model 1}
	\bY = \sum_{k=1}^{K}\bV\bH_k\bF^T\bX_k + \bZ,
\end{equation}
where $\bV$ is defined as
\begin{equation*}
	\bV \triangleq \bV_v \otimes \bV_h \in \bbC^{N_r \times F_vF_hN_r}, 
\end{equation*} 
$\bV_v\in\bbC^{N_{r,v} \times F_vN_{r,v}}$
and $\bV_h\in\bbC^{N_{r,h} \times F_hN_{r,h}}$ are partial discrete Fourier transformation (DFT) matrices, i.e.,
\begin{equation*}
	\bV_v = \tilde{\bI}_{N_{r,v}\times F_vN_{r,v}}\tilde{\bV}_v,  \ \bV_h = \tilde{\bI}_{N_{r,h}\times F_hN_{r,h}}\tilde{\bV}_h,
\end{equation*}  
$\tilde{\bV}_v$ and $\tilde{\bV}_h$ are $F_vN_{r,v}$ and $F_hN_{r,h}$ dimensional DFT matrices, respectively, $\tilde{\bI}_{N\times FN}$ is a matrix containing the first $N$ rows of the $FN$ dimensional identity matrix,  
$F_v$ and $F_h$ are two fine (oversampling) factors,
$\bF$ is defined as 
\begin{equation*}
	\bF  \triangleq \tilde{\bI}_{N_p\times F_\tau N_p}\tilde{\bF}\tilde{\bI}_{F_\tau N_p\times F_\tau N_f} 	\in \bbC^{N_p \times N_\tau N_f},
\end{equation*}
$\tilde{\bF}$ is the $F_\tau N_p$ dimensional DFT matrix, $\tilde{\bI}_{F_\tau N_p\times F_\tau N_f}$ is a matrix containing the first $F_\tau N_f$ columns of the $F_\tau N_p$ dimensional identity matrix, 
i.e., $\bF$ is the matrix obtained by $\tilde{\bF}$ after row extraction and column extraction,
\begin{equation*}
	N_f \triangleq \lceil {N_pN_g}/{N_c}\rceil,
\end{equation*}
$F_\tau$ is also a fine factor,
$\bH_k \in \bbC^{F_vF_hN_r \times F_\tau N_f}$ is the beam domain channel matrix of user $k$,
whose components follow the independent
complex Gaussian distributions with zero mean and possibly different variances, 
the diagonal matrix $\bX_k$ is the training signal of user $k$ satisfying
\begin{equation*}
	\bX_k^H\bX_k = \bI,
\end{equation*} 
and $\bZ$ is the noise matrix whose components are independent and identically distributed (i.i.d.) complex
Gaussian random variables with zero mean and variance $\sigma_z^2$. 
For the notational convenience, let
\begin{equation}
	\bF_d \triangleq \tilde{\bI}_{N_p\times F_\tau N_p}\tilde{\bF},
\end{equation}
and we have 
\begin{equation}
	\bF = \bF_d\tilde{\bI}_{F_\tau N_p\times F_\tau N_f}.
\end{equation}
From the definitions, we can obtain that
\begin{equation*}
	\bV_v\bV_v^H = F_vN_{r,v}\bI, 
\end{equation*} 
\begin{equation*}
	\bV_h\bV_h^H = F_hN_{r,h}\bI,
\end{equation*}
\begin{equation*}
	\bF_d\bF_d^H = F_\tau N_p\bI.
\end{equation*}
Denote the power matrix of beam domain channel as 
\begin{equation}\label{equ:Omega}
	\bOmega_k \triangleq \Exp\braces{\bH_k\odot \bH_k^*}, k\in\setposi{K}.
\end{equation}
Due to the channel sparsity, most of the components in $\bOmega_k$ are (close to) zero \cite{IGA,channelaqyou}. 
Then, \eqref{rece signal model 1} can be rewritten as
\begin{equation}
	\bY = \bV\bH\bM + \bZ,
\end{equation}
where $\bH = \left[ \bH_1, \bH_2,   \cdots , \bH_K \right] \in \bbC^{ F_vF_hN_r\times KF_\tau N_f }$ and $\bM \!=\! \left[ \bX_1\bF,  \bX_2\bF, \cdots, \bX_K\bF \right]^{T} \!\!\in\! \bbC^{KF_\tau N_f \times N_p}$. 
After vectorization, we have
\begin{equation}\label{equ:maMIMO rece signal 3}
	\by = \tilde{\bA}\tilde{\bh} + \bz,
\end{equation}
where $\by, \bz\in\bbC^{N\times 1}$ and  $\tilde{\bh}\in\bbC^{\tilde{M}\times 1}$ are the vectorizations of $\bY$, $\bZ$ and $\bH$, respectively, $\tilde{\bA}\triangleq \bM^T\otimes\bV \in\bbC^{N\times \tilde{M}}$, $N =N_rN_p$, $\tilde{M} = KF_aF_\tau N_rN_f$, and $F_a = F_vF_h$.
Since most components in $\tilde{\bh}$ are zero, we reduce the dimension of variables by extracting non-zero components.
Let
\begin{equation*}
	\bomega \triangleq \mtxvec{\left[ \bOmega_1 , \bOmega_2 , \cdots , \bOmega_K \right]},
\end{equation*} 
and $M \triangleq \lVert \bomega \rVert_0$ be the number of components in $\tilde{\bh}$ with non-zero variance, where $\norm{\cdot}_0$ is the $\ell_0$ norm.
Let the indexes of non-zero components in $\bomega$ be $\mathcal{P} \triangleq \braces{p_1,p_2,\ldots,p_M}$,
where $1\le p_1 < p_2 <\ldots <p_M \le \tilde{M}$. 
Define the column extraction matrix as 
\begin{equation}\label{equ:E}
	\bE  \triangleq \left[ \be_{p_1},  \be_{p_2},  \ldots,  \be_{p_M} \right]\in \bbC^{\tilde{M}\times M},
\end{equation} where $\be_i\in\bbC^{\tilde{M}\times 1}, i\in\mathcal{P}$, is the $i$-th column of the $\tilde{M}$ dimensional identity matrix. 
\eqref{equ:maMIMO rece signal 3} can be re-expressed as
\begin{equation}\label{equ:maMIMO rece signal 4}
	\by = \bA\bh + \bz,
\end{equation}
where $\bA = \tilde{\bA}\bE \in\bbC^{N\times M}$ is the matrix of $\tilde{\bA}$ after column extraction, $\bh\in\bbC^{M\times 1} = \bE^T\tilde{\bh}$ is the vector of $\tilde{\bh}$ after variable extraction.
In \eqref{equ:maMIMO rece signal 4}$, \bh \sim\mathcal{CN}\left(\mathbf{0},\bD\right)$ with diagonal and positive definite $\bD \triangleq \Diag{\bE^T\bomega}$ and $\bz \sim \mathcal{CN}\left(\mathbf{0},\sigma_z^2\bI\right)$.
We then have the following theorem.
\begin{theorem}\label{lemma:SR in mMIMO-CE 1}
	For matrix $\bA$ in \eqref{equ:maMIMO rece signal 4}, we have, 
	\begin{equation}
		\rho\left(N\bI - \bA^H\bA\right) \le\left(  KF_vF_hF_\tau-1\right)  N.
	\end{equation}
	In this case, if 
	\begin{equation}
	d < \frac{2}{KF_vF_hF_\tau},	
	\end{equation}
	then SIGA converges.
\end{theorem}
\begin{proof}
	See in Appendix \ref{proof:SR in mMIMO-CE 1}.
\end{proof}
For the case with $K = 48$, $M = 29277$, and $F_v= F_h= F_\tau=2$,
when general pilot sequences with constant magnitude
property are adopted, $d < 0.0052$ is sufficient to ensure the convergence of SIGA.
Note that this range is much larger than the worst case  $ d< \frac{2}{M} = 6.8\times 10^{-5}$ in Corollary \ref{corol:1}.
We finally consider the special case, where the adjustable phase shift pilots (APSPs) \cite{channelaqyou}, are used.
In this case, $\bA$ is equal to $\bA_p$ of \cite[Equation (42)]{SIGA}, and we have 
\begin{equation}\label{equ:aux2 in text}
	\bA = \tilde{\bA}_p\bE_p,
\end{equation}
where 
\begin{equation*}
	\tilde{\bA}_p = \bF_d\otimes \bV \in \bbC^{N \times F_vF_hF_\tau N},
\end{equation*}
and $\bE_p \in \bbC^{F_vF_hF_\tau N \times M}$ is another column extraction matrix similar with $\bE$ in \eqref{equ:E}, whose detailed definition can be found in \cite[above Equation (42)]{SIGA}.
\begin{theorem}\label{lemma:SR in mMIMO-CE 2}
	For $\bA$ in \eqref{equ:aux2 in text}, we have, 
	\begin{equation}
		\rho\left(N\bI - \bA^H\bA\right) \le  \left(F_vF_hF_\tau -1\right)N.
	\end{equation}
	In this case, if 
	\begin{equation}
		d < \frac{2}{F_vF_hF_\tau},
	\end{equation}
	then SIGA converges.
\end{theorem}
\begin{proof}
	See in Appendix \ref{proof:SR in mMIMO-CE 2}.
\end{proof}
{For the case with $F_v = F_h = F_\tau = 2$,
$d < 0.25$ is sufficient for SIGA to converge.}

\section{Simulation Results}
In this section, we present numerical simulations to illustrate the theoretical results.
This section is divided into two parts:
the first one focuses on the results for the general case, and the other one focuses on the massive MIMO-OFDM channel estimation.

\subsection{General Case}
For the general case, we first generate a matrix $\bA \in \bbC^{N\times M}$ with components drawn from i.i.d. $\mathcal{CN}\left(0,1\right)$ and then normalize its components so that their magnitude is $1$.
The dimension of $\bA$ is set to as $N = 300$ and $M = 150$.
In this experiment, the components of $\bh$ are drawn from i.i.d. $\mathcal{CN}\left(0,1\right)$, and the noise vector $\bz$ is a realization of white complex Gaussian noise.
The observation $\by$ is generated according to the received signal model $\by = \bA\bh + \bz$.
The variance of the noise $\bz$ is chosen to $\sigma_z^2 = 0.15$, and this value could achieve an SNR of  $30$ dB, where the SNR is defined as
$\textrm{SNR} \triangleq \Exp\braces{\norm{ \bA\bh }^2}/\Exp\braces{\norm{\bz}^2}$ in this experiment.
We focus on the convergence performance of $\bbnu$ as well as $\btheta$ in Part II of this paper.
A detailed analysis of the estimation performance of SIGA as well as numerical simulations can be found in \cite{SIGA}.

We first present the convergence performance of $\bbnu$ for different initializations as well as different damping factors.
To show the convergence of all components of $\bbnu$, we use $\norm{\bbnu}_2$ as a metric.
The initialization is set to $\bbnu\left(0\right) = \bzero$, $\bbnu\left(0\right) = \tilde{\bg}_{\textrm{min}} =  -\frac{N-1}{\sigma_z^2}\mathbf{1}$, and $\bbnu\left(0\right) =-1000\times \mathbf{1}$, respectively.
The damping factor is set to $d = 1$ and $d = 0.6$, respectively.
From Fig. \ref{Fig:random nu}, we can find $\norm{\bbnu}_2$ converges to the same fixed point in all settings and the fixed point satisfies $0<\norm{\bbnu^\star}_2 < \norm{\tilde{\bg}_\textrm{min}}_2$, which is in exact agreement with Theorem \ref{The-theta_1 conv}.
Note that the vertical axis of the image uses logarithmic coordinates, which results in $\ln(0) = -\infty$ cannot be drawn on the image.
We also find that the convergence rate of $\bbnu$ slows down with the decrease of the damping factor.

\begin{figure}[htbp]
	\centering
	\includegraphics[width=0.5\textwidth]{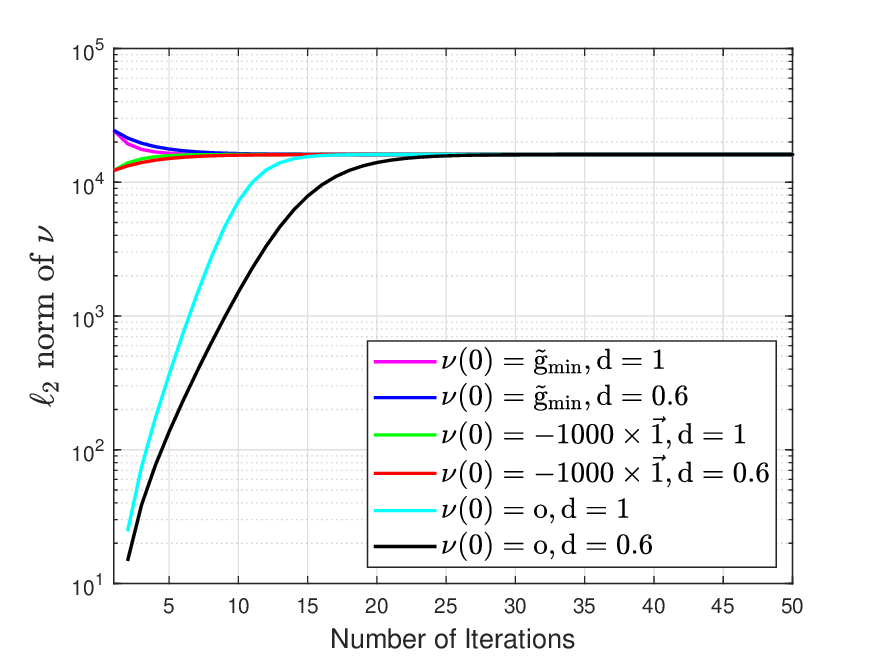}
	\caption{\small Convergence of $\norm{\bbnu\left(t\right)}$ for different initializations and damping factors in general case.}
	\label{Fig:random nu}
\end{figure}

\begin{figure}[htbp]
	\centering
	\includegraphics[width=0.5\textwidth]{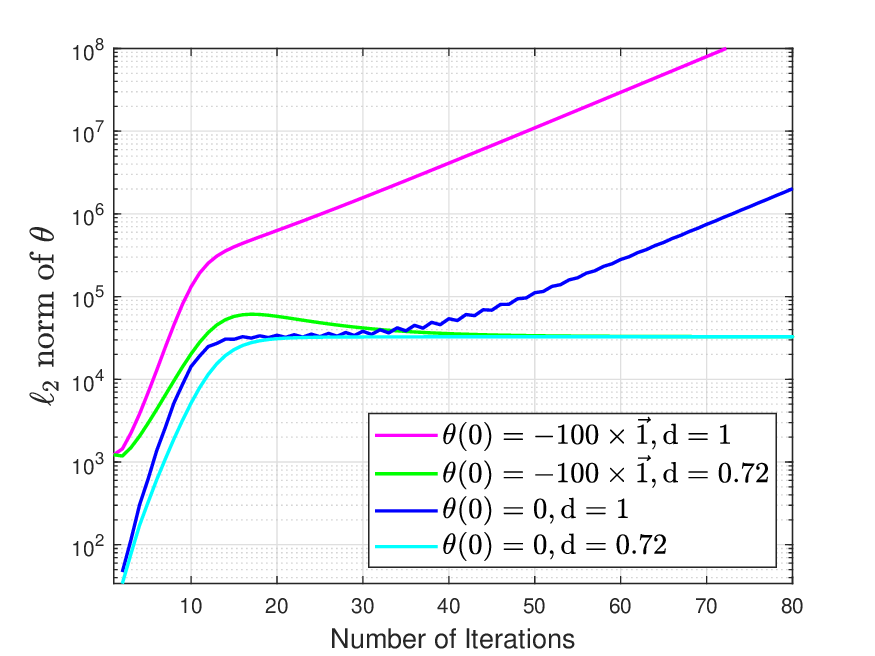}
	\caption{\small Convergence and divergence of $\norm{\btheta\left(t\right)}$ for different initializations and damping factors in general case.}
	\label{Fig:random theta}
\end{figure}

We next show the convergence performance of $\btheta$.
For $\bA$ in this experiment, we calculate $\rho\left( N\bI - \bA^H\bA \right) = 528.4643$.
According to Theorem \ref{corol:rang of v'}, $d < 0.7242$ is sufficient for the convergence of $\btheta$. 
We set the damping factor to $d = 1$ and $d = 0.72$, respectively.
For simplicity, the initialization of $\bbnu$ is set to $\bbnu\left(0\right) = \bzero$.
The initialization of $\btheta$ is set to $\btheta\left(0\right) = \bzero$ and $\btheta\left(0\right) = -100\times \mathbf{1}$, respectively.
From Fig. \ref{Fig:random theta}, we can find that $\norm{\btheta}_2$ converges to the same fixed point for all settings except for damping factor $d = 1$ (in this case, $\norm{\btheta}_2$ diverges).
This is exactly consistent with Theorem \ref{corol:rang of v'}.

\subsection{Massive MIMO Channel Estimation}
In this subsection, we focus on the convergence of $\bbnu$ and $\btheta$ in massive MIMO channel estimation.
The system parameter settings  as well as the acquisition of $\bh$ are detailed as described in \cite{IGA,SIGA}.
We summarize the system parameter settings in Table \ref{tab:para}.
\begin{table}[htbp]
	\centering
	\caption{Parameter Settings}\label{tab:para}
	\begin{tabular}{cc}
		\hline
		Parameter &Value \\
		\hline
		Number of BS antenna $N_{r,v}\times N_{r,h}$ & $8\times 16$ \\
		UT number $K$ & $48$ \\
		Center frequency $f_c$ & $4.8$GHz \\
		Number of training subcarriers $N_p$ & $360$ \\
		Subcarrier spacing $\Delta_f$ & $15$kHz \\
		Number of subcarriers $N_c$ & $2048$ \\
		CP length $\dnnot{N}{g}$ & $144$ \\
		Fine Factors $F_v, F_h, F_\tau$ & $2, 2, 2$ \\
		Mobile velocity of users   &   $3 - 10$kmph\\
		\hline
	\end{tabular}
\end{table}
Consider two types of pilots: the first is the general constant magnitude pilot and the other is APSPs.
They correspond to the cases described in Theorem \ref{lemma:SR in mMIMO-CE 1} and Theorem \ref{lemma:SR in mMIMO-CE 2}, respectively.

We first consider the general constant magnitude pilot.
In this case, the pilots of different users in \eqref{rece signal model 1} are generated as $\bX_k = \diag{\bx_k}, k\in \setposi{K}$, where the components of $\bx_k$ are drawn from i.i.d. $\mathcal{CN}\left(0,1\right)$ and then normalized to the unit magnitude.
Then, $\bA$ is generated as that in \eqref{equ:maMIMO rece signal 4}.
In this experiment, the dimension of $\bA$ is calculated as $N = 46080$ and $M = 29277$.
The variance of the noise $\bz$ is chosen to $\sigma_z^2 = 0.01$, and this value could achieve an SNR of  $20$ dB, where the SNR is defined as
$\textrm{SNR} \triangleq \frac{1}{\sigma_z^2}$ in this experiment \cite{SIGA}.
We then present the convergence performance of $\bbnu$.
The initialization of $\bbnu$ is set to be $\bbnu\left(0\right) = \bzero$, $\bbnu\left(0\right) = \tilde{\bg}_{\textrm{min}} = -\frac{N-1}{\sigma_z^2}\mathbf{1}$, and $\bbnu\left(0\right) = -\mathbf{1}$, repectively. The damping factor is the same as the previous experiment.
From Fig. \ref{Fig:ZC nu}, it can be found that $\norm{\bbnu}_2$ converges to the same fixed point in all settings, where the fixed point satisfies $0<\norm{\bbnu^\star}_2 < \norm{\tilde{\bg}_\textrm{min}}_2$. 
This observation is in exact agreement with Theorem \ref{The-theta_1 conv}.

\begin{figure}[htbp]
	\centering
	\includegraphics[width=0.5\textwidth]{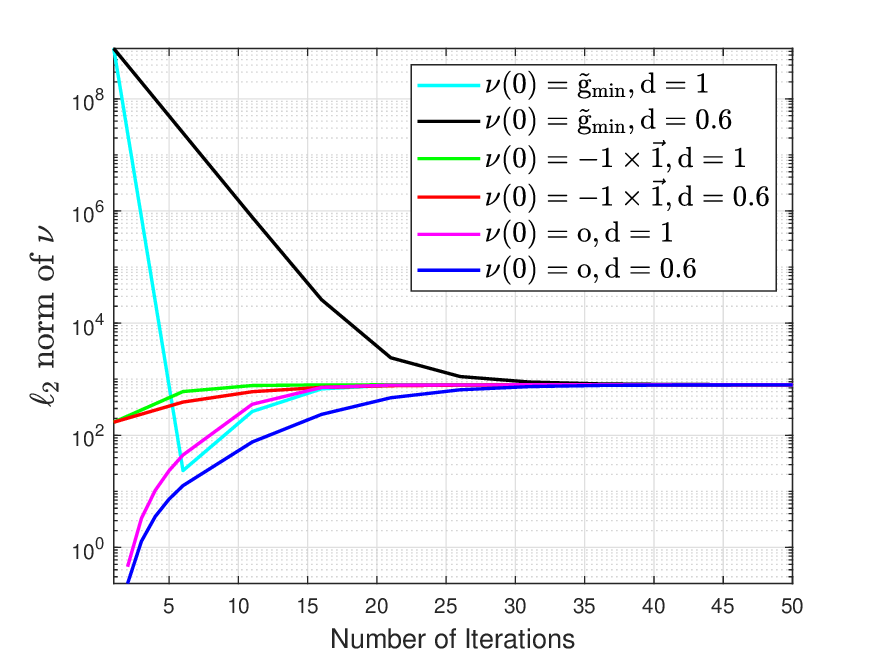}
	\caption{\small Convergence of $\norm{\bbnu\left(t\right)}$ for different initializations and damping factors for general constant magnitude pilot.}
	\label{Fig:ZC nu}
\end{figure}

\begin{figure}[htbp]
	\centering
	\includegraphics[width=0.5\textwidth]{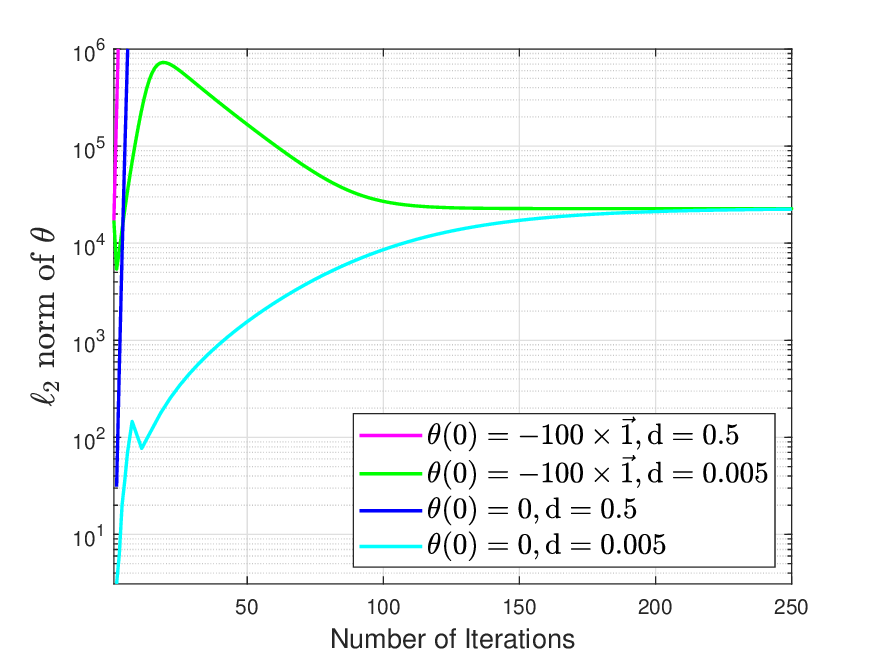}
	\caption{\small Convergence and divergence of $\norm{\btheta\left(t\right)}$ for different initializations and damping factors for general constant magnitude pilot.}
	\label{Fig:ZC theta}
\end{figure}

We now show the convergence performance of $\btheta$.
Due to the large dimension of $\bA$ in this experiment, the computational cost of $\rho\left( N\bI - \bA^H\bA \right)$ is relatively high.
We verify the range of damping factor in Theorem \ref{lemma:SR in mMIMO-CE 1}.
The damping factor is set to be $d = 0.5$ and $d = 0.005$, respectively.
Form Theorem \ref{lemma:SR in mMIMO-CE 1}, $d < 0.0052$ is sufficient to ensure the convergence of $\btheta$ in this experiment.
The initialization of $\bbnu$ is set to be $\bbnu\left(0\right) = \bzero$, and the initialization of $\btheta$ is set to be $\btheta\left(0\right) = \bzero$ and $\btheta\left(0\right) = -100\times \mathbf{1}$, respectively.
From Fig. \ref{Fig:ZC theta}, it can be found that $\norm{\btheta}_2$ diverges when $d = 0.5$ and converges to the same fixed point in case of $d = 0.005$.
This is consistent with Theorem \ref{lemma:SR in mMIMO-CE 1}.

\begin{figure}[htbp]
	\centering
	\includegraphics[width=0.5\textwidth]{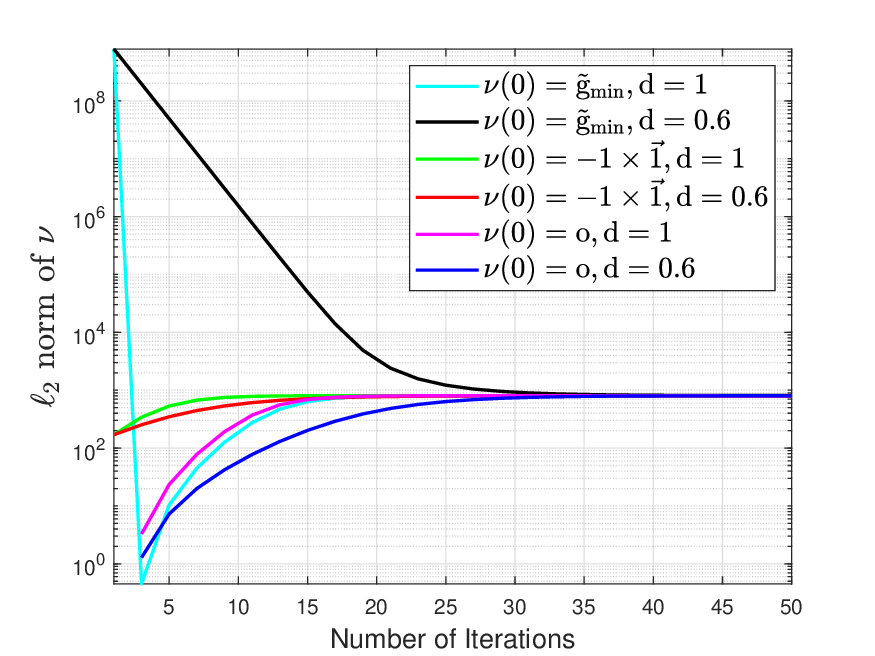}
	\caption{\small Convergence of $\norm{\bbnu\left(t\right)}$ for different initializations and damping factors for APSPs.}
	\label{Fig:APSP nu}
\end{figure}

\begin{figure}[htbp]
	\centering
	\includegraphics[width=0.5\textwidth]{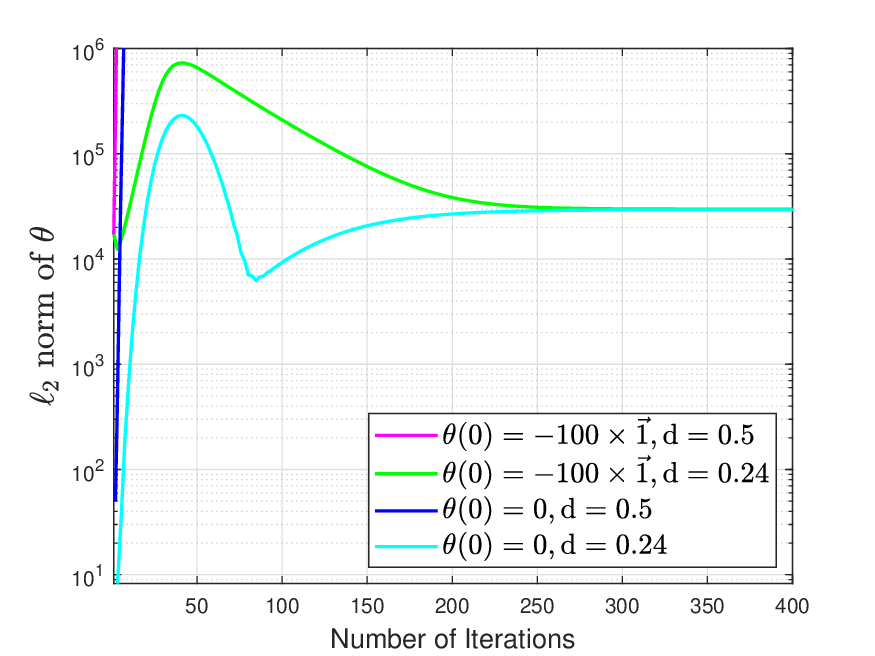}
	\caption{\small Convergence and divergence of $\norm{\btheta\left(t\right)}$ for different initializations and damping factors for APSPs.}
	\label{Fig:APSP theta}
\end{figure}

We finally consider the APSPs.
The APSP for the user $k$ is set to be $\bX_k = \Diag{\br\left(n_k\right)}\bP$, where
\begin{equation}
	\begin{split}
			&\br\left(n_k\right) \\
			= &\left[\exp\braces{-\barjmath 2\pi\frac{n_kN_1}{F_\tau N_p}}, \cdots, \exp\braces{-\barjmath 2\pi\frac{n_kN_2}{F_\tau N_p}} \right]^T\in \bbC^{N_p},
	\end{split}
\end{equation} 
$n_k \in \braces{0, 1, \cdots, F_\tau N_p-1}$ is the phase shift scheduled for the user $k$, and $\bP = \Diag{\bp}$ is the  basic pilot satisfying $\bP\bP^H = \bI$.  
Given the channel power matrix $\bOmega_k, k\in \setposi{K}$, we can use \cite[Algorithm 1]{channelaqyou} to determine the value of $n_k$ and thus $\bX_k, k\in \setposi{K}$.
Then, the matrix $\bA$ is generated as that in \eqref{equ:aux2 in text} (details can be found in \cite{SIGA}).
In this experiment, the dimension of $\bA$ is calculated as $N = 46080$ and $M = 29277$.
The noise variance $\sigma_z^2$ as well as the SNR are the same as in the previous experiment.
The convergence performance of $\bbnu$ is plotted in Fig. \ref{Fig:APSP nu} and similar conclusions can be drawn, where the initialization of $\bbnu$ and the damping factor are the same as in the previous experiment.
We then show the convergence performance of $\btheta$ in Fig. \ref{Fig:APSP theta}, where the damping factor is set to be $d = 0.5$ and $d = 0.24$,  respectively, and the initialization of $\bbnu$ and $\btheta$ is set the same as in the previous experiment.
Form Theorem \ref{lemma:SR in mMIMO-CE 2}, $d < 0.25$ is sufficient to ensure the convergence of $\btheta$ in this experiment.
It can be found that $\norm{\btheta}_2$ diverges when $d = 0.5$ and converges to the same fixed point in case of $d = 0.24$.
This is consistent with Theorem \ref{lemma:SR in mMIMO-CE 2}.

\section{Conclusion}
We have investigated the convergence of SIGA for massive MIMO-OFDM channel estimation. We analyze the convergence of SIGA  for a general Bayesian inference problem with the measurement matrix of constant magnitude property, and then apply the  general theories to the case of massive MIMO-OFDM channel estimation.
Through revisiting its iteration, we find that the iterating system of the common SONP is independent of that of the common FONP.
Hence, we can check the convergence of the common SONP separately.
It is then proved that given the initialization satisfying a particular and large range, the common SONP is convergent no matter the size of the damping factor.
For the convergence of the common FONP, we establish a sufficient condition. 
More specifically, we find that the convergence of the common FONP depends on the spectral radius of the iterating system matrix $\tilde{\bB}^\star$.
On this basis, we  establish condition for the damping factor that guarantees the convergence of $\btheta\left(t\right)$ in the worst case. 
{Further, we determine the range of the damping factor for massive MIMO-OFDM channel estimation by using the specific properties of the measurement matrices.}
Simulation results confirm the theoretical results.

\appendices

\section{Calculation of \eqref{equ:update of SIGA}}\label{proof:calculation of update}
We first show that $\bbnu\left(t+1\right)$ in \eqref{equ:tilde nu_s} can be re-expressed as that in \eqref{equ:update of nu_0}.
From \eqref{equ:tilde nu_s},  \eqref{equ:aux0 in AppA} on the next page  is direct.
\begin{figure*}
	\begin{equation}\label{equ:aux0 in AppA}
		\bbnu\left(t+1\right) = \underbrace{d\left(N-1\right)\left( \diag{\bD^{-1}-\left( \bLambda\left(\bbnu\left(t\right)\right) - \frac{1}{\beta\left(\bbnu\left(t\right)\right)}\bLambda^2\left(\bbnu\left(t\right)\right)\right)^{-1}} - \bbnu\left(t\right) \right)}_{d\bg\left(\bbnu\left(t\right)\right)}  + \left(1-d\right)\bbnu\left(t\right)
	\end{equation}
	\hrule
\end{figure*}
Thus,  $\bg\left(\bbnu\left(t\right)\right)$ can be expressed as \eqref{equ:aux1 in AppA} on the next page, where $\left(\textrm{a}\right)$ comes from that \eqref{equ:Lambda}.
\begin{figure*}
	\begin{equation}\label{equ:aux1 in AppA}
		\begin{split}
  &\bg\left(\bbnu\left(t\right)\right) = \left(N-1\right) \diag{ \bD^{-1} - \left( \bLambda\left(\bbnu\left(t\right)\right) - \frac{1}{\beta\left(\bbnu\left(t\right)\right)}\bLambda^2\left(\bbnu\left(t\right)\right) \right)^{-1}} - \left(N-1\right)\bbnu\left(t\right)\\
  =&\left(N-1\right)\diag{ \bD^{-1} - \Diag{\bbnu\left(t\right)} - \left( \bLambda\left(\bbnu\left(t\right)\right) - \frac{1}{\beta\left(\bbnu\left(t\right)\right)}\bLambda^2\left(\bbnu\left(t\right)\right) \right)^{-1}}\\
  \equaa&\left(N-1\right)	\diag{ \bLambda^{-1}\left(\bbnu\left(t\right)\right) - \bLambda^{-1}\left(\bbnu\left(t\right)\right)\left( \bI - \frac{1}{\beta\left(\bbnu\left(t\right)\right)}\bLambda\left(\bbnu\left(t\right)\right) \right)^{-1} } = -\left(N-1\right)\diag{\left(\beta\left(\bbnu\left(t\right)\right)\bI - \bLambda\left(\bbnu\left(t\right)\right)\right)^{-1}}
	\end{split}
	\end{equation}
	\hrule
\end{figure*}
We then show that when $t = 0$, the matrices that need to be inverted in \eqref{equ:aux1 in AppA} are intertible.
From \eqref{equ:Lambda} and $\bbnu\left(0\right) \le \bzero$, we can obtain that $\bLambda\left(\bbnu\left(0\right)\right)$ is positive definite and hence invertible.
From \eqref{equ:beta}, we have
\begin{equation*}
	\beta\left(\bbnu\left(0\right)\right) > \left[ \bLambda\left(\bbnu\left(0\right)\right) \right]_{i,i} > 0, i\in \setposi{M}.
\end{equation*}
This implies that
\begin{equation*}
	\beta\left(\bbnu\left(0\right)\right)\bI - \bLambda\left(\bbnu\left(0\right) \right)
\end{equation*}
is positive definite and hence invertible.
Moreover, combining \eqref{equ:aux0 in AppA} and \eqref{equ:aux1 in AppA}, we have $\bg\left(\bbnu\left(0\right)\right) < \bzero$, and 
\begin{equation*}
	\bbnu\left(1\right) = d\bg\left(\bbnu\left(0\right)\right) + \left(1-d\right)\bbnu\left(0\right) < \bzero
\end{equation*}
is finite.
Following by that, assuming that at the $t$-th iteration, where $t \ge 1$, we have
$\bbnu\left(t\right) < \bzero$ is finite, $\bLambda\left(\bbnu\left(t\right)\right)$ and
\begin{equation*}
	\beta\left(\bbnu\left(t\right)\right)\bI - \bLambda\left(\bbnu\left(t\right) \right)
\end{equation*} 
are positive definite and invertible.
In the same way, it can be readily checked that $\bbnu\left(t+1\right) < \bzero$ is finite.
Hence, we have $\bLambda\left(\bbnu\left(t+1\right)\right)$ and
\begin{equation*}
	\beta\left(\bbnu\left(t+1\right)\right)\bI - \bLambda\left(\bbnu\left(t\right) \right)
\end{equation*}
are positive definite and invertible.
By induction, we have shown that when $t \ge 1$, we have $\bbnu\left(t\right) < \bzero$ is finite,
and for $t \ge 0$, $\bLambda\left(\bbnu\left(t\right)\right)$ and 
\begin{equation*}
	\beta\left(\bbnu\left(t\right)\right)\bI - \bLambda\left(\bbnu\left(t\right) \right)
\end{equation*}
are positive definite and invertible.

We now show that $\btheta\left(t+1\right)$ in \eqref{equ:tilde theta_s} can be re-expressed as that in \eqref{equ:update of theta_0}.
From \eqref{equ:tilde theta_s}, we can obtain  \eqref{equ:aux2 in AppA}, where we omit some of the counter $t$ at the right hand side of the equation for the notational convenience, and
\begin{equation}\label{equ:T}
	\bT\left(\bbnu\right) = \left( \bI - \frac{1}{\beta\left(\bbnu\right)} \bLambda\left(\bbnu\right) \right)^{-1},
\end{equation} 
where the matrix invertibility comes from \eqref{equ:Lambda and beta} directly. 
\begin{figure*}
	\begin{equation}\label{equ:aux2 in AppA}
		\begin{split}
		&\btheta\left(t+1\right) = 
			\frac{d\left(N-1\right)}{N}\bT\left(\bbnu\right)\left[ \frac{1}{\beta\left(\bbnu\right)}\bA^H\left(2\by\right) - \frac{1}{\beta\left(\bbnu\right)}\bA^H \bA\bLambda\left(\bbnu\right){\btheta}\left(t\right)  + N\btheta\left(t\right) \right] + \left(1-dN\right)\btheta\left(t\right)\\
			=&\underbrace{\frac{2d\left(N-1\right)}{\beta\left(\bbnu\right) N}\bT\left(\bbnu\right)\bA^H\by}_{\bb\left(\bbnu\right)} + \underbrace{\left[ \frac{d\left(N-1\right)}{N}\bT\left(\bbnu\right)\left(N\bI - \frac{1}{\beta\left(\bbnu\right)}\bA^H\bA\bLambda\left(\bbnu\right)  \right)\btheta\left(t\right) - d\left(N-1\right)\btheta\left(t\right) \right]}_{d\bB\left(\bbnu\right) \btheta\left(t\right) }  + \left(1-d\right)\btheta\left(t\right)
		\end{split}
	\end{equation}
	\hrule
\end{figure*}
Thus, we can obtain 
\begin{align}
	&\bB\left(\bbnu\right)  \\
	=&  \frac{\left(N-1\right)}{N}\bT\left(\bbnu\right)\left(N\bI - \frac{1}{\beta\left(\bbnu\right)}\bA^H\bA\bLambda\left(\bbnu\right)  \right) - \left(N-1\right)\bI \nonumber \\
	=& \left(N-1\right) \left[\frac{\bT\left(\bbnu\right)}{N}\left( N\bI - \frac{\bA^H\bA\bLambda\left(\bbnu\right)}{\beta\left(\bbnu\right)} \right)  - \bT\left(\bbnu\right)\bT^{-1}\left(\bbnu\right) \right]                     \nonumber \\
	=&  \frac{\left(N-1\right)}{\beta\left(\bbnu\right)} \bT\left(\bbnu\right)\left( \bI - \frac{1}{N}\bA^H\bA \right)\bLambda\left(\bbnu\right).               \nonumber
\end{align}
Also, it is not difficult to show that given a finite $\btheta\left(0\right)$, $\btheta\left(t\right)$ is finite at each iteration.
This completes the proof.

\section{Proof of Lemma \ref{lemma:function tilde{g}}}\label{proof:function tilde{g}}
From Appendix \ref{proof:calculation of update}, it can be checked that given  $\bbnu \le  \bzero$, $\tilde{\bg}\left(\bbnu\right)$ and $\bg\left(\bbnu\right)$ are well defined.
Denote $g_i\left(\bbnu\right)$, $\nu_{i}$, $d_i$ and ${\lambda}_i\left(\bbnu\right)$ as the $i$-th components of $\bg\left(\bbnu\right)$, $\bbnu$, the diagonals of  $\bD$ and $\bLambda\left(\bbnu\right)$, respectively, where $i\in \setposi{M}$.
Due to $\bbnu \le \bzero$, we have
\begin{subequations}
		\begin{equation}
		\beta\left(\bbnu\right) = \sigma_z^2 + \sum_{i=1}^{M}\lambda_{i}\left(\bbnu\right) {>}0,
	\end{equation}
	\begin{equation}\label{equ:lambda i}
		{\lambda}_{i}\left(\bbnu\right) = \frac{1}{d_i^{-1} - \nu_{i}}{>}0,
	\end{equation}
	\begin{equation}\label{g_i}
		\begin{split}
		   	g_{i}(\bbnu) &=- \frac{N-1}{\beta(\bbnu) - {\lambda}_i(\bbnu)} \\
		   	&= -\frac{N-1}{\sigma_z^2 + \sum_{i'\neq i}{\lambda}_{i'}\left(\bbnu\right)} < 0.	
		\end{split}
	\end{equation}
\end{subequations}
From \eqref{g_i} and \eqref{equ:update of nu_0}, the two properties of $\bg\left(\bbnu\right)$ and $\tilde{\bg\left(\bbnu\right)}$, i.e., the monotonicity and the scalability, are not difficult to see. 
We next show its boundedness. 

From the definitions, we have
\begin{subequations}
	\begin{equation}
		\lim_{\nu_{1},\nu_{ 2},\ldots,\nu_{ M}\rightarrow-\infty}\beta\left(\bbnu\right) = \sigma_z^2.
	\end{equation}
\end{subequations} 
From the monotonicity of $\bg\left(\bbnu\right)$, we can obtain
\begin{equation}
	\bg\left(\bbnu\right)>	\lim_{\nu_{1},\nu_{ 2},\ldots,\nu_{ M}\rightarrow-\infty}\bg\left(\bbnu\right) =  -\frac{N-1}{\sigma_z^2}\mathbf{1} = \tilde{\bg}_{\textrm{min}} .
\end{equation}
Thus,  $\tilde{\bg}_{\textrm{min}}<\bg\left(\bbnu\right)<\bzero$.
Then,  $\tilde{\bg}_{\textrm{min}}<\tilde{\bg}\left(\bbnu\right) < \bzero$ directly follows from the definition of $\tilde{\bg}\left(\bbnu\left(t\right)\right)$ in \eqref{equ:update of nu_0}. 
This completes the proof.


\section{Proof of Theorem \ref{The-theta_1 conv}}\label{proof:theta_1 conv}

Consider $\bbnu\left(1\right) = \tilde{\bg}\left(\bbnu\left(0\right)\right)$.
If $\bbnu\left(1\right) \le \bbnu\left(0\right)$, by Lemma \ref{lemma:function tilde{g}},
\begin{equation}
	\bbnu\left(2\right) = \tilde{\bg}\left(\bbnu\left(1\right)\right)\le \tilde{\bg}\left(\bbnu\left(0\right)\right) = \bbnu\left(1\right).
\end{equation}
Then, the sequence $\bbnu\left(t\right)$ is a decreasing sequence.
By Lemma \ref{lemma:function tilde{g}}, this sequence is also bounded.
Thus, it converges to a finite vector $\bbnu^\star$.
Also, by Lemma \ref{lemma:function tilde{g}}, we have the result of Theorem \ref{The-theta_1 conv}.
The case of $\bbnu\left(1\right) \ge \bbnu\left(0\right)$ can be similarly proved.
This completes the proof.

\section{Proof of Lemma \ref{The-synchronous updates}}\label{Proof-The-syn}
Define $\btheta^\star$ as
\begin{equation}
	\btheta^\star \triangleq \left(\bI - \tilde{\bB}^\star\right)^{-1}\bb^\star.
\end{equation}
Since $\rho\left(\tilde{\bB}^\star\right)<1$, $1$ is not an eigenvalue of $\tilde{\bB}^\star$ and $\bI-\tilde{\bB}^\star$ is invertible. Thus,  the above $\btheta^\star$ exists

We next show that $\btheta\left(t\right)$ converges to $\btheta^\star$. 
Since $\rho\left(\tilde{\bB}^\star\right) <1$, there exists a matrix norm $\norm{\cdot}$ such that \cite[Lemma 5.6.10]{horn2012matrix}
\begin{equation}
	\norm{\tilde{\bB}^\star} < 1.
\end{equation}
Then, let $\norm{\cdot}$ be the vector norm that induces the matrix norm $\norm{\cdot}$ \cite[Definition 5.6.1]{horn2012matrix}.
Define the error between $\btheta\left(t\right)$ and $\btheta^{\star}$ as
\begin{equation}
	\varepsilon\left(t\right) \triangleq \lVert \btheta\left(t\right) - \btheta^{\star} \rVert.
\end{equation}
Then, we can obtain \eqref{ineq0}  on the next page.
\begin{figure*}
	\begin{equation}\label{ineq0}
		\begin{split}
			&\varepsilon\left(t+1\right) = \lVert \btheta\left(t+1\right) - \btheta^{\star}\rVert
			= \lVert \tilde{\bB}\left( \bbnu\left(t\right) \right)\btheta\left(t\right) - \tilde{\bB}^{\star}\btheta^{\star} + \bb\left(\bbnu\left(t\right)\right)-\bb^{\star} \rVert\\
			&=\lVert \tilde{\bB}\left(\bbnu\left(t\right)\right)\left( \btheta\left(t\right)-\btheta^{\star}\right) + \left( \tilde{\bB}\left(\bbnu\left(t\right)\right)-\tilde{\bB}^{\star} \right)\btheta^{\star}  + \bb\left(\bbnu\left(t\right)\right)-\bb^{\star}  \rVert \\
			&{\le}\lVert \tilde{\bB}\left(\bbnu\left(t\right)\right)\left( \btheta\left(t\right)-\btheta^{\star}\right) \rVert  + \lVert \left( \tilde{\bB}\left(\bbnu\left(t\right)\right)-\tilde{\bB}^{\star} \right)\btheta^{\star} + \bb\left(\bbnu\left(t\right)\right)-\bb^{\star} \rVert\\
			&{\le }\norm{\tilde{\bB}\left(\bbnu\left(t\right)\right)}\varepsilon\left(t\right) + \lVert \left( \tilde{\bB}\left(\bbnu\left(t\right)\right)-\tilde{\bB}^{\star} \right)\btheta^{\star} + \bb\left(\bbnu\left(t\right)\right)-\bb^{\star} \rVert
		\end{split}
	\end{equation}
	\hrule
\end{figure*}
Define a sequence $c\left(t\right)$ as
\begin{equation}\label{definition-ct}
	c\left(t\right)\triangleq\lVert \left( \tilde{\bB}\left(\bbnu\left(t\right)\right)-\tilde{\bB}^{\star} \right)\btheta_0^{\star} + \bb\left(\bbnu\left(t\right)\right)-\bb^{\star} \rVert.
\end{equation}
Since $\bbnu$ converges to $\bbnu^\star$, we have
\begin{equation*}
	\liminfty{t}\tilde{\bB}\left(\bbnu\left(t\right)\right) = \tilde{\bB}^{\star}, \ \liminfty{t}\bb\left(\bbnu\left(t\right)\right) = \bb^{\star},
\end{equation*}
and thus,
\begin{equation}\label{equ:aux1 in App2}
		\liminfty{t}c\left(t\right)= 0, 
\end{equation}
\begin{equation}\label{equ:aux2 in App2}
	\liminfty{t}\norm{\tilde{\bB}\left(\bbnu\left(t\right)\right)} = \norm{\tilde{\bB}^\star} <1.
\end{equation}
Let 
\begin{equation}
	\delta_1 \triangleq \frac{1 - \norm{\tilde{\bB}^\star}}{2} >0,
\end{equation}
\begin{equation}
	\delta_2 \triangleq \norm{\tilde{\bB}^\star} + \delta_1 <1.
\end{equation}
To show 
\begin{equation*}
	\liminfty{t}\varepsilon\left(t\right) = 0,
\end{equation*}
we only need to show that $\forall \epsilon > 0$, $\exists \  t_0$, when $t > t_0$, we have
\begin{equation*}
	\varepsilon\left(t\right) < \epsilon.
\end{equation*} 
From \eqref{equ:aux1 in App2} and \eqref{equ:aux2 in App2}, we can obtain that $\exists \ t_1 > 0$, when $t > t_1$, we have
\begin{equation*}
	c\left(t\right) < \frac{\epsilon\left(1-\delta_2\right)}{2},
\end{equation*}
\begin{equation*}
	\norm{\tilde{\bB}\left(\bbnu\left(t\right)\right)} \le \delta_2.
\end{equation*}
Then, for $t \ge t_1 $, we have
\begin{equation*}
	\varepsilon\left(t+1\right) \le \delta_2\varepsilon\left(t\right) + \frac{\epsilon\left(1-\delta_2\right)}{2},
\end{equation*} 
and hence for any positive integer $\Delta t$,
\begin{equation*}
	\begin{split}
		&\varepsilon\left(t+\Delta t\right)\\ < & \delta_2^{\Delta t}\varepsilon\left(t\right) + \left(\delta_2^{\Delta t-1} + \delta_2^{\Delta t-2} + \cdots + \delta_2^{0}\right)  \frac{\epsilon\left(1-\delta_2\right)}{2} \\
		<& \delta_2^{\Delta t}\varepsilon\left(t\right) + \frac{\epsilon}{2}.
	\end{split}
\end{equation*}
Since $0<\delta_2<1$, we have 
\begin{equation*}
	\liminfty{\Delta t}\delta_2^{\Delta t} = 0.
\end{equation*}
Let $\Delta t$ such that 
\begin{equation*}
	\delta_2^{\Delta t} < \frac{\epsilon}{2\varepsilon\left(t_1\right)}.
\end{equation*}
Let $t_0 = t_1 + \Delta t$.
Then, when $t> t_0$, we have
\begin{align*}
	\varepsilon\left(t\right) &= \varepsilon\left(t_0 + t - t_0\right) = \varepsilon\left(t_1 + \Delta t + t - t_0\right) \\
	&< \delta_2^{ \Delta t + t - t_0}\varepsilon\left(t_1\right) + \frac{\epsilon}{2} < \delta_2^{\Delta t}\varepsilon\left(t_1\right) + \frac{\epsilon}{2} \\ \nonumber
	&< \frac{\epsilon}{2\varepsilon\left(t_1\right)}\varepsilon\left(t_1\right) + \frac{\epsilon}{2} = \epsilon.
\end{align*}
This proves 
\begin{equation*}
	\lim\limits_{t\to \infty} \varepsilon\left(t\right) = 0.
\end{equation*}
Since all vector norms are equivalent, it implies that $\norm{\btheta\left(t\right) - \btheta^\star}_2$ with the Euclidean norm also goes to zero as $t \to \infty$. 
This completes the proof of Lemma \ref{The-synchronous updates}.

\section{Proof of Lemma \ref{Lemma-spectral radius}}\label{{Proof-Lemma-spec}}
From \eqref{equ:update of nu_0}, we have
\begin{align}
	&\bbnu^{\star} = \bg\left(\bbnu^\star\right) \nonumber \\
	=&-\left(N-1\right)\diag{\left( \beta^{\star}\bI - \bLambda^{\star}    \right)^{-1}}.
\end{align}
From the definition of $\beta^{\star}$ in \eqref{equ:beta0 star} and $\bLambda^{\star}$ in \eqref{equ:Lambda0 star}, we can readily show that $\beta^{\star}\bI - \bLambda^{\star} $ is invertible.
Since we have proven that $\bbnu^{\star} <\bzero$ in Theorem \ref{The-theta_1 conv},
from the definition of $\bLambda^\star$ in \eqref{equ:Lambda0 star}, we can obtain
\begin{equation}\label{Lambda *}
	\begin{split}
		\bLambda^\star &= \left( \bD^{-1} -\Diag{\bbnu^{\star}} \right)^{-1} \\
		&{\prec} \left( -\Diag{\bbnu^{\star}} \right)^{-1}
		{=}\frac{1}{N-1}\left( \beta^*\bI - \bLambda^* \right).
	\end{split}
\end{equation}
From the definition, we can obtain $\bLambda^\star$ is diagonal positive definite.
Let $\lambda_i^\star = \left[ \bLambda^\star \right]_{i,i}, i\in \setposi{M}$.
Hence, $\lambda_i^\star$ is an eigenvalue of $\bLambda^{\star}$ and  $\lambda_i^\star > 0, i\in \setposi{M}$. 
Then, from \eqref{Lambda *}, 
we have 
\begin{equation}
	\lambda_i^\star - \frac{\beta^\star - \lambda_i^\star}{N-1} < 0, i\in \mathcal{Z}_M^+,
\end{equation}
which implies that $\lambda_i^\star < \frac{\beta^*}{N}, i\in \mathcal{Z}_M^+$. Hence, we have $\rho\left( \bLambda^\star \right) < \frac{\beta^\star}{N}$.  This completes the proof.

\section{Proof of Lemma \ref{lemma:eigens of B star}}\label{proof:eigens of B star}
We first prove that the eigenvalues of $\bB^\star$ are all real.
Let
\begin{equation}
	\begin{split}
		\bQ &\triangleq \left( \bI - \frac{1}{N}\bD^{-1}\bLambda^\star \right)^{1/2}\left( \frac{\bLambda^\star}{\beta^\star} \right)^{1/2}\left( N\bI - \bA^H\bA \right)\\
		& \ \ \times\left( \bI - \frac{1}{N}\bD^{-1}\bLambda^\star \right)^{1/2}\left( \frac{\bLambda^\star}{\beta^\star} \right)^{1/2}\\
		&= \bK^{-1}\bB^\star\bK \sim \bB^\star,
	\end{split}
\end{equation}
where $\bK$ is the following diagonal positive definite matrix: 
\begin{equation}
	\bK = \left( \frac{\bLambda^\star}{\beta^\star}\right)^{-1/2}\left( \bI - \frac{1}{N}\bD^{-1}\bLambda^\star \right)^{1/2}. 
\end{equation}
Thus, $\bB^\star$ and $\bQ$ have the same eigenvalues.
From the definition, $\bQ$ is Hermitian. Therefore, the eigenvalues of $\bQ$ and $\bB^\star$ are all real. 
Then, from  \eqref{equ:Q in Appendix} on the next page,
\begin{figure*}
	\begin{equation} \label{equ:Q in Appendix}
			\bQ = \underbrace{N\left(\bI - \frac{1}{N}\bD^{-1}\bLambda^\star\right)\frac{\bLambda^\star}{\beta^\star}}_{\bQ_1} +  \underbrace{\left( \bI - \frac{1}{N}\bD^{-1}\bLambda^\star \right)^{1/2}\left( \frac{\bLambda^\star}{\beta^\star} \right)^{1/2}\left( - \bA^H\bA \right) \left( \bI - \frac{1}{N}\bD^{-1}\bLambda^\star \right)^{1/2}\left( \frac{\bLambda^\star}{\beta^\star} \right)^{1/2}}_{\bQ_2}
	\end{equation}
	\hrule
\end{figure*}
we have $\bQ_1, \bQ_2$ are Hermitian, and hence \cite[6.70 (a), pp116]{seber2008matrix}
\begin{equation}
	\lambda_{max}\left(\bQ\right) \le \lambda_{max}\left(\bQ_1\right) + \lambda_{max}\left(\bQ_2\right).
\end{equation}
Then, for $\bQ_1$, we can readily check that it is positive definite, and thus
\begin{equation}
	\lambda_{max}\left(\bQ_1\right) = \rho\left(\bQ_1\right) \overset{\left(\mathrm{a}\right)}{\le}N\rho\left(\bI - \frac{1}{N}\bD^{-1}\bLambda^\star\right)\rho\left(\frac{\bLambda^\star}{\beta^\star}\right)
	\overset{\left(\mathrm{b}\right)}{<}1,
\end{equation}
where $\left(\mathrm{a}\right)$ comes from \cite[Exercise below Theorem 5.6.9]{horn2012matrix} and $\left(\mathrm{b}\right)$ comes from \eqref{equ:aux1 in text} and \eqref{equ:aux2 in text}. 
Define $\bK_1$ as
\begin{equation}
	\bK_1 = \left( \bI - \frac{1}{N}\bD^{-1}\bLambda^\star \right)^{1/2}\left( \frac{\bLambda^\star}{\beta^\star} \right)^{1/2}.
\end{equation}
Then, we can obtain that $-\bQ_2 = \left(\bA\bK_1\right)^H\bA\bK_1$.
Hence, $\bQ_2$ is negative semidefinite,  and we have $\lambda_{max}\left(\bQ_2\right) \le  0$. Thus, we have
\begin{equation}
	\lambda_{max}\left(\bQ\right) < 1.
\end{equation}
Since $\bB^\star \sim \bQ$, we have $\lambda_{max}\left(\bB^\star\right) < 1$.
Then, from \eqref{equ:range of eigen of B}, we have 
\begin{equation}
	\rho\left(\bB^{\star}\right)< 1\times \rho\left(N\bI - \bA^H\bA\right)\times \frac{1}{N} = \frac{\rho\left(N\bI - \bA^H\bA\right)}{N}.
\end{equation}
Thus, we can obtain that 
\begin{equation*}
	-\frac{\rho\left(N\bI - \bA^H\bA\right)}{N} < \lambda_{B,i} < 1.
\end{equation*}
This completes the proof.

\section{Proof of Theorem \ref{lemma:radius of NI-A^HA}}\label{proof:radius of NI-A^HA}
We first give the range of $\rho\left(\bA^H\bA\right)$.
From the definition, $\bA^H\bA$ is positive semidefinite. The eigenvalues $v_1\le v_2\le  \cdots\le  v_M$ of $\bA^H\bA$ are real and nonnegative. 
Thus, we can obtain $v_M = \rho\left(\bA^H\bA\right)$.
Then, we have
\begin{align}
	\rho\left(\bA^H\bA\right) &= v_M  \le \sum_{m=1}^{M}v_m 
	= \tr{\bA\bA^H} \nonumber\\
	&\le \sum_{m=1}^{M}v_M = M\rho\left( \bA^H\bA \right).
\end{align} 
When $\left|a_{i,j}\right|=1, \forall i,j$, we can obtain
\begin{equation}
	\tr{\bA\bA^H} = \lVert \bA\rVert_F^2 = NM.
\end{equation}
Thus, we have $\rho\left(\bA^H\bA\right) \le NM \le M\rho\left(\bA^H\bA\right)$, which implies that
\begin{equation}
	N\le \rho\left( \bA^H\bA \right)\le NM.
\end{equation}
Hence, we have $0\le v_1\le \cdots \le v_M \le NM$.
The eigenvalues of $N\bI - \bA^H\bA$ are $v_m' =N - v_m, m\in \setposi{M}$.
Thus, we have $N-NM \le v_m'\le N$, and $\left|  v_m'\right| \le \max\braces{N,NM-N}$.
Since in this paper $M >1$, we have $\rho\left(N\bI - \bA^H\bA\right) \le NM-N$.
If $\rank{\bA} = 1$, then $\bA$ can be decomposed as $\bA = \ba\bb^H$, where $\ba \in \bbC^{N\times 1}$, $\bb \in \bbC^{M\times 1}$, and $\ba$ and $\bb$ are non-zero.
Combining \cite[6.54 (c)]{seber2008matrix}, $\bb \bb^H$ and $\bb^H \bb$ are positive semi-definite, we can obtain that 
\begin{align}
	\rho\left(\bA^H\bA\right) &= \rho\left( \bb \ba^H\ba \bb^H \right) = \ba^H\ba\rho\left( \bb \bb^H \right) =\ba^H\ba\rho\left( \bb^H \bb \right)\nonumber  \\ 
	&= \tr{\bA^H\bA} = NM.
\end{align}
Then, we have $\rho\left(N\bI - \bA^H\bA\right) = NM -N$.
This completes the proof.

\section{Proof of Theorem \ref{lemma:SR in mMIMO-CE 1}}\label{proof:SR in mMIMO-CE 1}

From the definition, we have 
\begin{equation*}
	\bA^H\bA = \bE^T\left( \tilde{\bA}^H\tilde{\bA} \right)\bE,
\end{equation*}
which implies that $\bA^H\bA$ is a principal submatrix of $\tilde{\bA}^H\tilde{\bA}$.
Combining \cite[Theorem 4.3.28]{horn2012matrix} and the elementary transformation, we have
\begin{equation}
	\lambda_{max}\left(\bA^H\bA\right) \le\lambda_{max}\left(\tilde{\bA}^H\tilde{\bA}  \right),  
\end{equation}
which implies that $\rho\left(\bA^H\bA\right) \le \rho\left(\tilde{\bA}^H\tilde{\bA}\right)$.
From \cite[6.54 (c), pp 107]{seber2008matrix}, we can obtain $\rho\left(\tilde{\bA}^H\tilde{\bA}\right) = \rho\left(\tilde{\bA}\tilde{\bA}^H\right)$.
From the definition, we have
\begin{equation}
	\begin{split}
		\tilde{\bA}\tilde{\bA}^H &= \left(\bM^T \otimes \bV\right)\left( \bM^T \otimes \bV \right)^H \\
		&= \left( \bM^T\bM^* \right) \otimes \left( \bV_v\otimes \bV_h \right)\left( \bV_v^H\otimes \bV_h^H \right)\\
		&= F_vF_hN_r\bK \otimes \bI,
	\end{split}     
\end{equation}
where $\bK = \sum_{k=1}^{K}\bX_k\bF\bF^H\bX_k^H$. 
Since $\bK$ is Hermitian, we can decompose $\bK$ as $\bK = \bU\bLambda_K\bU^H$, where $\bU$ is unitary. 
Then, we can obtain
\begin{equation}
	\bK \otimes \bI = \left( \bU\otimes \bI \right)\left(\bLambda_K\otimes \bI\right)\left(\bU^H\otimes \bI\right) = \bU'\bLambda_K'\left(\bU'\right)^H,
\end{equation} 
where $\bU'$ is unitary and $\bLambda_K'$ is diagonal. 
Since $\bK \otimes \bI$ is also Hermitian, we can obtain that 
\begin{equation*}
	\rho\left(\bK\otimes \bI\right) = \rho\left(\bK\right).
\end{equation*}
Since $\bX_k$ is unitary, we also have $\rho\left(\bX_k\bF\bF^H\bX_k^H\right) = \rho\left( \bF\bF^H \right), \forall k$.
Finally, we have
\begin{equation}
	\begin{split}
		&\rho\left( \tilde{\bA}^H\tilde{\bA} \right) = \rho\left(\tilde{\bA}\tilde{\bA}^H\right) 
		= F_vF_hN_r\rho\left(\bK\right)\\
		\overset{\left(\mathrm{a}\right)}{\le}& F_vF_hN_r\sum_{k=1}^{K}\rho\left(\bF\bF^H\right) = KF_vF_hN_r\rho\left(\bF\bF^H\right),
	\end{split}
\end{equation}
where $\left(\textrm{a}\right)$ comes from \cite[6.70 (a), pp 116]{seber2008matrix} and $\bX_k\bF\bF^H\bX_k^H$ is positive semi-definite.
Similarly, we can obtain 
\begin{align}
	\rho\left( \bF\bF^H \right) &= \rho\left(\bF^H\bF\right) = \rho\left( 		 \tilde{\bI}_{F_\tau N_p\times F_\tau N_f}^T \bF_d^H\bF_d\tilde{\bI}_{F_\tau N_p\times F_\tau N_f} \right) \nonumber \\
	&\le \rho\left(  \bF_d^H\bF_d \right) = F_\tau N_P ,
\end{align}
and hence,
\begin{equation}
	\rho\left(\bA^H\bA\right)\le KF_vF_hF_\tau N_r N_p = KF_vF_hF_\tau N.
\end{equation}
From a similar process in Appendix \ref{proof:radius of NI-A^HA}, we can obtain 
\begin{equation}\label{equ:aux1 in appH}
	\rho\left(N\bI-\bA^H\bA\right) \le KF_vF_hF_\tau N - N.
\end{equation}
Substituting \eqref{equ:aux1 in appH} into the right hand side of \eqref{equ:range of d 0}, we have 
\begin{equation}
	\frac{2}{1 + \frac{\rho\left( N\bI - \bA^H\bA \right)}{N}} \ge \frac{2}{KF_vF_hF_\tau}.
\end{equation}
In this case, if $d < \frac{2}{KF_vF_hF_\tau}$, then SIGA converges.
This completes the proof.

\section{Proof of Theorem \ref{lemma:SR in mMIMO-CE 2}}\label{proof:SR in mMIMO-CE 2}
From the definitions, it is not difficult to obtain that 
\begin{align}
	\rho\left(\bA^H\bA\right) \le \rho\left(\tilde{\bA}_p^H\tilde{\bA}_p\right) = \rho\left(  \tilde{\bA}_p\tilde{\bA}_p^H\right)= F_vF_hF_\tau N. 
\end{align}
Hence, we can obtain that 
\begin{equation*}
	\rho\left( N\bI - \bA^H\bA \right) \le \left( F_vF_hF_\tau -1 \right)N. 
\end{equation*}
Similarly, substituting the above range into the right hand side of \eqref{equ:range of d 0}, we have 
\begin{equation}
	\frac{2}{1 + \frac{\rho\left( N\bI - \bA^H\bA \right)}{N}} \ge \frac{2}{F_vF_hF_\tau}.
\end{equation}
In this case, if $d < \frac{2}{F_vF_hF_\tau}$, then SIGA converges.
This completes the proof.

\bibliographystyle{IEEEtran}  
\bibliography{IEEEabrv,reference}

\end{document}